\documentclass[final,1p,times]{elsarticle}
\usepackage{geometry}
\usepackage{amsthm}
\usepackage{colortbl}
\usepackage{xcolor}

\usepackage{algpseudocode}
\usepackage{algorithm}
\usepackage{bm}
\usepackage{soul}
\usepackage{enumitem}
\usepackage{amssymb}
\usepackage{pifont}
\journal{J. Comp. Phys.}
\usepackage{amsfonts,amssymb,amsbsy,amsmath,tabulary,graphicx,times,caption,fancyhdr}
\usepackage[utf8]{inputenc}
\usepackage{url,morefloats,floatflt,cancel,tfrupee}
\usepackage{mathrsfs}
\usepackage{mathtools,relsize}
\usepackage{graphicx}
\DeclareMathOperator{\erfc}{erfc}

\newtheorem{theorem}{Theorem}[section]
\newtheorem{lemma}[theorem]{Lemma}

\newtheorem{remark}{Remark}[section]

\usepackage{hyperref}
\usepackage{setspace}
\usepackage{color}

\begin{document}

\begin{frontmatter}

\title{Random batch sum-of-Gaussians algorithm for molecular dynamics simulations of Yukawa systems in three dimensions}
\author[1]{Chen Chen}\ead{ronnie7@sjtu.edu.cn} 
\author[1,2]{Jiuyang Liang}\ead{liangjiuyang@sjtu.edu.cn}
\author[1,3]{Zhenli Xu}\ead{xuzl@sjtu.edu.cn}

\affiliation[1]{organization={School of Mathematical Sciences, Shanghai Jiao Tong University}, city={Shanghai 200240},	country={China}}

\affiliation[2]{organization={Center for Computational Mathematics, Flatiron Institute, Simons Foundation}, city={New York 10010}, country={USA}
}

\affiliation[3]{organization={MOE-LSC and CMA-Shanghai, Shanghai Jiao Tong University}, city={Shanghai 200240}, country={China}}

\begin{abstract}
Yukawa systems have drawn widespread interest across various applications, including plasma physics, colloidal science, and astrophysics, due to their critical role in modeling electrostatic interactions. In this paper, we introduce a novel random batch sum-of-Gaussians (RBSOG) algorithm for molecular dynamics simulations of three-dimensional Yukawa systems with periodic boundary conditions. We develop a sum-of-Gaussians (SOG) decomposition of the Yukawa kernel, dividing the interactions into near-field and far-field components. The near-field component, singular but compactly supported in a local domain, is  calculated directly. The far-field component, represented as a sum of smooth Gaussians, is treated using the random batch approximation in Fourier space with an adaptive importance sampling strategy to reduce the variance of force calculations. Unlike the traditional Ewald decomposition, which introduces discontinuities and significant truncation error at the cutoff, the SOG decomposition achieves high-order smoothness and accuracy near the cutoff, allowing for efficient and energy-stable simulations. Additionally, by avoiding the use of the fast Fourier transform, our method achieves optimal $O(N)$ complexity while maintaining high parallel scalability. Finally, unlike previous random batch approaches, the proposed adaptive importance sampling strategy achieves nearly optimal variance reduction across the regime of the coupling parameters, which is essential for handling varying coupling strengths across weak and strong regimes of  electrostatic interactions. Rigorous theoretical analyses are presented, including SOG decomposition construction, variance estimation, and simulation convergence.  We validate the performance of RBSOG method through numerical simulations of one-component plasma under weak and strong coupling conditions, using up to $10^6$ particles and 1024 CPU cores. As a practical application in fusion ignition, we simulate high-temperature, high-density deuterium-$\alpha$ mixtures to study the energy exchange between deuterium and high-energy $\alpha$ particles. Due to the flexibility of the Gaussian approximation, the RBSOG method can be readily extended to other dielectric response functions, offering a promising approach for large-scale simulations.

\end{abstract}







\begin{keyword}
Molecular dynamics simulations \sep Yukawa systems \sep sum-of-Gaussians decomposition \sep adaptive importance sampling




\end{keyword}

\end{frontmatter}
\section{Introduction}
Plasma systems have garnered significant attention due to their importance in various applications such as nuclear fusion, the stability of magnetic confinement devices, and microelectronic materials~\cite{zylstra2022burning, kritcher2022design,killian2007ultracold,graziani2012large}. The Yukawa potential \( Y(r)=e^{-r/\lambda}/r \), derived from linearly screened theory, serves as a fundamental force field in molecular dynamics (MD) simulations~\cite{Frenkel2001Understanding} which is one of popular tools for the study of plasma physics as well as chemical physics and biophysics. Because of its exponential decay form, the Yukawa potential is often overlooked and crudely treated as a short-range kernel in mainstream MD softwares  \cite{thompson2022lammps, abraham2015gromacs}. However, in practical systems like warm dense plasmas and burning plasmas used in inertial fusion~\cite{zylstra2022burning}, both temperature and number density can become extremely high. This causes the Yukawa kernel to exhibit large-$\lambda$ behavior similar to that of a long-range kernel and leads to a huge number of interaction neighbors, making direct truncation impractical. 

Numerous fast algorithms have been developed for classical Coulomb systems, corresponding to the $\lambda\rightarrow\infty$ case. Most of these algorithms can be categorized into two groups: fast Fourier transform (FFT)-accelerated Ewald summation methods~\cite{Darden1993JCP,Hockney1988Computer,essmann1995smooth,lindbo2011spectral}, and adaptive tree-based methods such as the fast multipole method (FMM)~\cite{greengard1987fast,cheng1999fast,greengard2002new,fmm1,fmm2,fmm3,fmm5} and the tree code method~\cite{Barnes1986Nature,li2009cartesian}, achieving complexity of $O(N)$ or $O(N\log N)$. Extending these methods to the case of finite $\lambda$ requires finding a suitable decomposition or expansion of the Yukawa potential. For instance, such extensions have been achieved by FFT-based methods using the generalized Ewald decomposition~\cite{dharuman2017generalized}, by FMM using plane-wave expansions~\cite{greengard2002new,huang2009fmm} or modified spherical Bessel expansions~\cite{boschitsch1999fast}, and by the tree code using appropriate Cartesian Taylor expansions~\cite{li2009cartesian}.  

Despite the remarkable achievements of these algorithms, especially in biophysical simulations~\cite{lu2007new,CHEN2018750}, they may encounter limitations in MD simulations of plasmas: When $\lambda$ is large, the long-range interactions with periodic boundary conditions necessitate the inclusion of numerous image boxes  within the interaction range in real space for the FMM, substantially increasing memory usage and computational cost. While periodic FMMs~\cite{jiang2023jcp,yan2018flexibly} mitigate this issue, it persists as the periodic tiling is divided into a smooth far part and a near part containing its nearest neighboring cells. The use of communication-intensive FFT leads to increased communication latency in parallel computing for large-scale systems which constrain both spatial and temporal scales~\cite{ayala2021scalability}. Moreover, the Ewald splitting/multipole expansion exhibits discontinuity at the near-field cutoff/boundary of leaf nodes~\cite{toxvaerd2012energy,Shamshirgar2019JCP}, resulting in a significant truncation error and noticeable energy drift, particularly at high temperatures. These combined effects pose significant challenges for achieving efficient and accurate simulations of plasma systems~\cite{graziani2012large}.

Recently, a class of stochastic algorithms has emerged, namely the random batch Ewald (RBE) method~\cite{Jin2020SISC,liang2022npt}. The RBE method employs importance sampling in the Fourier space, achieving a mathematically optimal linear complexity among Ewald-type algorithms. A recent improvement involves the use of the u-series decomposition~\cite{DEShaw2020JCP}, where the far-field kernel is represented as a series of Gaussians. This approach effectively addresses the inherent discontinuity issue and gives rise to the random batch sum-of-Gaussians (RBSOG) method~\cite{Liang2023SISC}. While these methods have demonstrated good scalability in all-atom simulations~\cite{liang2022superscalability,gao2024rbmd}, they remain limited to the pure Coulomb case. Extending the framework to Yukawa systems  poses significant challenges for two main reasons. First, the u-series decomposition used in the RBSOG is only applicable to power functions $r^{-\beta}$ with $\beta>0$~\cite{DEShaw2020JCP}, and a high-order regularity decomposition for the Yukawa kernel is still lacking. Second, even with such a decomposition, the importance sampling schemes used in the RBE and RBSOG methods are only effective at the Coulomb limit and the charge neutrality condition is satisfied. Addressing these challenges for Yukawa systems is particularly non-trivial due to the lack of scale invariance and one has to handle diverse scenarios across a broad range of screening and coupling strengths.%

In this paper, we propose a fast and adaptive random batch sum-of-Gaussians method for efficient and accurate simulations of fully periodic Yukawa systems in three dimensions, perfectly addressing aforementioned issues. We develop a novel SOG decomposition applicable to the Yukawa kernel for arbitrary $\lambda\in(0,\infty)$, achieving high-order regularity and uniform error across the entire interaction range. This decomposition is constructed by employing the truncated trapezoidal rule~\cite{Trefethen2014SIAMREV} to discretize the integral expression of the Yukawa kernel,
\begin{equation}\label{eq::ftr}
	Y(r)=\int_{-\infty}^{\infty}f(t,r)dt,\quad \text{with}\quad f(t,r):=\dfrac{1}{\sqrt{\pi}}e^{-r^2e^t-\frac{1}{4\lambda^2e^t}+t/2},
\end{equation}
such that the starting point and quadrature weights of the trapezoidal rule can be finely adjusted to ensure high-order smoothness. The feasibility of this construction scheme is guaranteed by rigorous error estimates. Subsequently, we apply the idea of random batch sampling to Fourier space calculations to achieve an optimal $O(N)$ scaling similar to previous RBE and RBSOG methods, but with a newly-proposed adaptive importance sampling strategy, inspired by the theory of Debye-H$\ddot{\text{u}}$ckel limit~\cite{hansen2013theory,hu2022symmetry}. We prove that this strategy achieves near-optimal variance reduction in the sense of ensemble averaging, due to its accurate representation of the structure factor at long-wave modes.


The resulting RBSOG method for Yukawa systems offers several advantages, demonstrated by rigorous error
analysis and systematic experiments. It achieves $O(N)$ complexity, reducing both calculation and communication cost. In comparison with the particle-particle particle-mesh (PPPM) method~\cite{Hockney1988Computer,dharuman2017generalized} and parallel volume fast multipole method (PVFMM)~\cite{malhotra2015pvfmm} for large-scale simulations with $1.28\times10^6$ particles on $1024$ cores, the RBSOG improves the performance over one order of magnitude, with parallel scalability remaining $90\%$. Since the RBSOG method is tree- and mesh-free, it also reduces memory usage by a factor of $40\%$. As a practical application in nuclear fusion ignition~\cite{zylstra2022burning}, we simulate deuterium-$\alpha$ mixtures at temperatures up to $3.48\times 10^7\,K$ and number densities up to $45.2$ particles per $\,\mathring{A}^{-3}$. The RBSOG method accurately captures energy exchange between deuterium and high-energy $\alpha$ particles while maintaining energy stability for at least $10^7$ simulation steps. This method exhibits broad applicability in MD simulations of Yukawa systems and can be seamlessly extended to other kernels used in plasma simulations through integration with the kernel-independent SOG method~\cite{greengard2018anisotropic,gao2021kernelindependent}. 
The main contributions of this paper can be highlighted as follows:
\begin{enumerate}[label=(\arabic*)]
\item We propose an SOG decomposition of the Yukawa kernel, resulting in near-field and far-field components, which are efficiently handled in real and Fourier spaces, respectively. The high-order smoothness of the decomposition significantly improves the accuracy of force calculations, thus greatly reducing energy drifts in practical simulations.
\item Our method avoids the use of communication-intensive framework FFT, instead we develop random batch method for the far-field calculation in Fourier space to achieve linear complexity and better parallel scalability. Theoretical analysis is provided to demonstrate the accuracy of the method.
\item Periodic
FMMs require dividing the periodic tiling into far-field and near-field regions containing immediate neighboring cells. The RBSOG avoids this difficulty and operates directly on the fundamental cell, resulting in better performance with simpler implementation.
\item Our method is applicable to general Yukawa systems with $\lambda\in (0,\infty)$. Compared to the sampling strategies employed in these earlier random batch methods, the adaptive importance sampling scheme in our approach achieves a significant variance reduction of $2-4$ times across a wide range of coupling parameters.
\end{enumerate}

The remaining part of this paper is structured as follows. Section~\ref{Sec::2.1} provides a brief review of basic linear response theory and the derivation of the Yukawa potential. In Section~\ref{Sec::2}, we introduce a novel SOG decomposition of the Yukawa kernel along with its error analysis. Section~\ref{sec::fast} provides a detailed description and analysis of the RBSOG algorithm. Simulation results are presented in Section~\ref{sec::numexa}. Concluding remarks are made in Section~\ref{sec::conclu}.

\section{Linear response theory and the Yukawa potential}\label{Sec::2.1}

Consider a charged system of $N$ particles located at $\{\bm{r}_i= (x_i,y_i,z_i),~i=1,\cdots,N\}$ with charge $\{q_i,~i=1,\cdots, N\}$ in a cuboid domain $\Omega$ with side lengths $L_x$, $L_y$, and $L_z$, respectively, in the presence of a background charge
density $n_{\text{b}}(\bm{r})$ (e.g., electrons) of $\bm{r}\in\Omega$. Given the charge distribution, the electrostatic potential $\Phi(\bm{r})$ satisfies the following Poisson equation:
\begin{equation}\label{eq::1.0.1}
	-\Delta \Phi(\bm{r})=4\pi\left[\sum_{i=1}^{N}q_{i}\delta(\bm{r}-\bm{r}_i)-n_{\text{b}}(\bm{r})\right].
\end{equation}
Applying the convolution theorem to Eq.~\eqref{eq::1.0.1}, one can represent $\Phi(\bm{r})$ in Fourier space as 

\begin{equation} \widehat{\Phi}(\bm{k})=-\frac{4\pi}{k^2}\left[\rho(-\bm{k})-\widehat{n}_{\text{b}}(\bm{k})\right]. 
\end{equation} 
where $\rho(\bm{k}):=\sum_{j=1}^{N}q_{j}e^{i\bm{k}\cdot\bm{r}_j}$ denotes the structure factor and $k=|\bm{k}|$. Assuming the instantaneous evolution of the background density with the point particles provides a reasonable approximation for light species, such as electrons. One can approximate the structure of the background as a linear response to the potential fluctuations arising from the presence of the point particles, allowing to express $\widehat{n}_{\text{b}}$ in terms of a charge response function $\xi(\bm{k})$ as
$
\widehat{n}_{\text{b}}(\bm{k})=\xi(\bm{k})\rho(-\bm{k}).
$
One thus writes the electrostatic potential as
\begin{equation} \widehat{\Phi}(\bm{k})=-\frac{4\pi}{k^2}\epsilon(\bm{k})^{-1}\rho(-\bm{k}), 
\end{equation} 
where
\begin{equation}
	\epsilon(\bm{k}):=\frac{1}{1+4\pi\xi(\bm{k})/{k^2}}
\end{equation}
represents the dielectric response function. When $\epsilon(\bm{k})\equiv 1$, the system exhibits no response and reduces to the pure Coulomb system. A widely used form of $\epsilon(\bm{k})$ in plasma simulations is
\begin{equation}\label{eq::epsilon}
	\epsilon(\bm{k})=1+(\lambda k)^{-2},
\end{equation}
where $\lambda$ is the screening length associated with the background. The associated potential is given by 
\begin{equation}
	\Phi(\bm{r})=\sum_{i=1}^{N} q_iY(|\bm{r}-\bm{r}_{i}|),
\end{equation}
which is the sum of Yukawa potential, also known as the screened Coulomb kernel in many applications. It should be noted that other more complex dielectric response functions may arise in plasma simulations, such as the exact gradient-corrected screening form \cite{stanton2015unified}:
\begin{equation}
	\epsilon_{\text{EGS}}(\bm{k})=1+(\lambda k)^{-2}(1+\nu \lambda^2k^2/4)^{-1},
\end{equation}
where the $k^{-4}$ term corresponds to either electronic correlations or Heisenberg uncertainty, and $\nu$ is a parameter characterizing the strength of gradient correction in the free energy. This type of potential can be written in terms of partial fractions as the difference between two Yukawa potentials, so that existing algorithms developed for the Yukawa kernel can be seamlessly extended to this potential. 

In practical plasma systems such as ultracold neutral plasmas, dusty plasmas, and those at the National Ignition Facility, the range of $\lambda$ can vary significantly~\cite{zylstra2022burning}. When the screening length $\lambda$ is large, the exponential decay is not obvious and a simple cutoff method is inefficient due to the large number of image charges within the neighbor list. Recent advancements by Dharuman \emph{et al.}~\cite{dharuman2017generalized,silvestri2022sarkas} explore the integration of Ewald splitting of Yukawa kernel with the PPPM framework similar to the Coulomb kernel, highlighting their potential to mitigate this issue for systems with moderate density and scales. However, for large-scale simulations, the spatial and parallel scalability of their method are still limited by the high communication cost in the FFT and memory usage for storing Fourier grids. Another drawback of the Ewald decomposition is the lack of smoothness around the cutoff point, which leads to significant truncation errors and various undesirable artifacts in force calculations~\cite{hammonds2020shadow}.

\section{Sum-of-Gaussians decomposition of Yukawa kernel}\label{Sec::2}
 
In this section, we first develop an SOG approximation for the Yukawa kernel with a uniform error bound on a specified domain $[\delta, R]$ with $\delta$ and $R>0$. We then demonstrate how this approximation can be used to construct a high-order regularity decomposition for the Yukawa potential and discuss the associated decomposition error. It is noted that  constructing an SOG approximation is a well-studied problem~\cite{greengard2018anisotropic,gao2021kernelindependent,beylkin2005approximation,beylkin2010approximation,trefethen2006talbot,xu2013bootstrap,AAMM-13-1126}. An optimal SOG approximation can often be obtained by using generalized Gaussian quadratures~\cite{ma1996generalized} to discretize the integral representation similar to Eq.~\eqref{eq::ftr}. For the Yukawa kernel, this approach is not straightforward because an optimal quadrature should depend on parameter $\lambda$ in the kernel. Recent efforts~\cite{greengard2018anisotropic,gao2021kernelindependent} to enable kernel-independent construction of SOG approximations also face this limitation. Therefore, we seek an SOG approximation that can be computed on the fly at negligible cost.

\subsection{Sum-of-Gaussians approximation}\label{sec::SOGapproximation}

We begin by discretizing the integral representation in Eq.~\eqref{eq::ftr} using the trapezoidal rule, which achieves spectral accuracy due to the exponential decay of the integrands at $\pm \infty$, and then truncating it at appropriate terms. 
We then establish a rigorous error estimate to ensure that only a moderate number of Gaussians is included in the approximation. Moreover, these new theoretical results play a crucial role in constructing the SOG decomposition described in Section~\ref{sec::2.3}. 


Before presenting the main theorem, we introduce some useful lemmas. 





\begin{lemma}\label{lemma::1.1}
Let $\alpha\geq 1$ and $0<\beta\leq 1$, then $g(x)=\alpha e^{-\beta x^2}$ is an upper bound function of $\erfc(x)$ for $x\geq 0$. 
\end{lemma}

\begin{lemma}\label{lemma::1.3}
The Fourier transform of $f(t,r)$ given in Eq.~\eqref{eq::ftr} with respect to variable $t$ is given by
\begin{equation}
\widehat{f}(k,r)=\frac{2}{\sqrt{\pi}}(2\lambda r)^{-\frac{1}{2}+2\pi i k}K_{\frac{1}{2}-2\pi i k}\left(\frac{r}{\lambda}\right),
\end{equation}
where 
\begin{equation}\label{eq::Knux}
K_{\nu}(x)=\frac{x^{\nu}}{2^{\nu+1}}\int_{0}^{\infty}e^{-t-\frac{x^2}{4t}}t^{-(\nu+1)}dt
\end{equation}
denotes the modified Bessel function of the second kind with order $\nu$ \emph{\cite{abramowitz1964handbook}}.

\end{lemma}
\begin{lemma}\label{lemma::1.4}
For real $k$, finite and positive $\lambda$, and $r\in[\delta,R]$, the following inequality holds:
	\begin{equation}
		\left|K_{\frac{1}{2}-2\pi ik}\left(\frac{r}{\lambda}\right)\right|\leq \frac{\pi\lambda^{1/2}}{r^{1/2}}e^{-\pi^2k+\frac{1}{6|1-2\pi i k|}}.
	\end{equation}
\end{lemma}
The proof of Lemma \ref{lemma::1.1} can be simply done by taking the first-order derivative of $\mathcal{G}(x)=\alpha e^{-\beta x^2}-\erfc(x)$ and studying the monotonicity. 
The proof of Lemma \ref{lemma::1.3} proceeds by sequentially following these steps: applying the definition of Fourier transform to $f(t,r)$, performing a change of variable $u=r^2e^t$, and using the representation given in Eq.~\eqref{eq::Knux}. The proof of Lemma \ref{lemma::1.4} is provided in~\ref{app::A.1}. These lemmas lead to the proof of the main result, Theorem ~\ref{thm::1.1}. 


\begin{theorem}\label{thm::1.1}
Let $\varepsilon$ be a prescribed tolerance, $[\delta, R]$ be the approximation range, and $\lambda>0$. Let $t_m:=t_0+mh$ with $t_0$ being referred to as the starting point of trapezoidal rule throughout this paper. For any $\delta\in(0,R)$, there exist a step size $h$, a real number $t_0\in[0,h)$ and two integers $M_1$ and  $M_2$
such that
\begin{equation}\label{eq::1.2}
\left|\frac{e^{-r/\lambda}}{r}-h\sum_{m=-M_1}^{M_2}f\left(t_m, r\right)\right|\leq \varepsilon,\quad r\in[\delta, R],
\end{equation}
where $f$ is defined via Eq.~\eqref{eq::ftr}.
\end{theorem}

\begin{proof}[Proof]
By the Poisson summation formula, one has
\begin{equation}\label{eq::2.1}
\sum_{m \in \mathbb{Z}} \widehat{f}\left(\frac{m}{h},r\right) e^{2 \pi i t_0 \frac{m}{h}}=h \sum_{m \in \mathbb{Z}} f\left(t_m,r\right).
\end{equation}
Truncating the right-hand side of Eq.~\eqref{eq::2.1} at $-M_1\leq m\leq M_2$ and rearranging terms, one obtains the SOG approximation of the Yukawa kernel:
\begin{equation}
\dfrac{e^{-r/\lambda}}{r}\approx h\sum_{m=-M_1}^{M_2}f\left(t_m,r \right).
\end{equation}
The error can be estimated by 
     \begin{equation}\label{eq::1.7}
     	\left|\dfrac{e^{-r/\lambda}}{r}-h\sum_{m=-M_1}^{M_2}f\left(t_m,r\right)\right|\leq E_{T}+E_{A},
     \end{equation}
 	 where the truncation error $E_T$ and the ``aliasing'' error $E_A$ are given by
 	 \begin{equation}
 	 	E_T=h\left|\sum_{m=-\infty}^{-M_1-1}f\left(t_m,r \right)+\sum_{m=M_2+1}^{\infty}f\left(t_m,r\right)\right|,
 	 \end{equation}
    and
    \begin{equation}\label{eq::1.9.}
    E_A=\left|\dfrac{e^{-r/\lambda}}{r}-h\sum_{m\in\mathbb{Z}}f\left(t_m,r\right)\right|=\left|\sum_{\substack{m \in \mathbb{Z}\\m\neq 0}} \widehat{f}\left(\frac{m}{h},r\right) e^{2 \pi i t_0 \frac{m}{h}}\right|,
    \end{equation}
    respectively. The truncation error $E_T$ satisfies
\begin{equation}
    \begin{split}
       E_T\leq \frac{h}{\sqrt{\pi}}\left(\left|\sum_{m=-\infty}^{-M_1-1}e^{\frac{t_m}{2}-\frac{1}{4\lambda^2e^{t_{m}}}-\delta^2e^{t_{m}}}\right|+\frac{1}{\sqrt{\pi}}\left|\sum_{m=M_2+1}^{\infty}e^{\frac{t_m}{2}-e^{t_m}r^2}\right|\right):=E_{T,\text{1}}+E_{T,\text{2}}
    \end{split}
\end{equation}
by the definition of $f(t,r)$, where $E_{T,\text{1}}$ and $E_{T,\text{2}}$ are the lower and upper tails of $E_{T}$, respectively. By leveraging the monotonic decrease of the function $e^{\frac{t}{2}-\frac{1}{4\lambda^2e^{t}}-\delta^2e^{t}}$ over the interval $t\in(-\infty,t_{-M_1-1}]$ and estimating $e^{-t}/(4\lambda^2)$ by its upper bound, the lower tail $E_{T,\text{1}}$ is bounded by
\begin{equation}\label{eq::1.9}
	\begin{split}
	E_{T,\text{1}}\leq \frac{1}{\sqrt{\pi}}\int_{-\infty}^{t_{-M_1}}e^{\frac{y}{2}-\frac{1}{4\lambda^2}e^{-y}-\delta^2e^y}dy\leq \frac{e^{-\frac{1}{4\lambda^2}e^{-t_{-M_1}}}}{\sqrt{\pi}}\int_{-\infty}^{t_{-M_1}}e^{\frac{y}{2}-\delta^2e^y}dy=\frac{e^{-\frac{1}{4\lambda^2}e^{-t_{-M_1}}}}{\delta}\left[1-\frac{\Gamma\left(\frac{1}{2},\delta^2e^{t_{-M_1}}\right)}{\Gamma\left(\frac{1}{2}\right)}\right]
	\end{split}
\end{equation}
where 
\begin{equation}\label{eq::21}
	\Gamma(\beta,x)=\int_{x}^{\infty}e^{-s}s^{\beta-1}ds
\end{equation}
represents the incomplete Gamma function. Note that the result in the right-hand side of Eq.~\eqref{eq::1.9} works for all $r\in[\delta,R]$. By Eqs.~\eqref{eq::1.9}-\eqref{eq::21}, it follows that 
\begin{equation}\label{eq::1.11}
	\begin{split}
	E_{T,\text{1}}\leq \frac{e^{-\frac{1}{4\lambda^2}e^{-t_{-M_1}}}}{\sqrt{\pi}\delta}\int_{0}^{\delta^2e^{t_{-M_1}}}y^{-1/2}e^{-y}dy\leq \frac{e^{-\frac{1}{4\lambda^2}e^{-t_{-M_1}}}}{\sqrt{\pi}\delta}\int_{0}^{\delta^2e^{t_{-M_1}}}y^{-1/2}dy= \frac{2e^{\frac{t_{-M_1}}{2}-\frac{1}{4\lambda^2}e^{-t_{-M_1}}}}{\sqrt{\pi}}
	\end{split}
\end{equation}
where the second inequality is obtained by bounding $e^{-y}$ by $1$. Similarly, one can estimate the upper tail $E_{T,2}$ by  
\begin{equation}\label{eq::1.12}
	E_{T,2}\leq \frac{1}{\sqrt{\pi}}\int_{t_{M_2}}^{\infty}e^{y/2-r^2e^y}dy=\frac{1}{r}\erfc(re^{\frac{t_{M_2}}{2}}) 
\end{equation}
due to the fact that the function $e^{t/2-e^{t}r^2}$ is monotonally decreasing on $t\in [t_{M_2+1},\infty)$ for any $M_2$ satisfing 
\begin{equation}\label{eq::t_m2}
t_{M_2+1}\geq -2\log(\sqrt{2}R)
\end{equation}
and $r\in[\delta,R]$. Applying Lemma~\ref{lemma::1.1} with $\alpha=\beta=1$, one has
\begin{equation}\label{eq::ET2}
	E_{T,2}\leq\frac{1}{r}e^{-r^2e^{t_{M_2}}}\leq \frac{1}{\delta}e^{-\delta^2e^{t_{M_2}}}.
\end{equation}
For the aliasing error $E_A$, applying Lemma \ref{lemma::1.3} gives  
\begin{equation}\label{eq::2.29}
	\begin{split}
		E_A\leq \left|\sum_{\substack{m \in \mathbb{Z}\\m\neq 0}}\frac{2}{\sqrt{\pi}}(2\lambda r)^{-\frac{1}{2}+2\pi i \frac{m}{h}}K_{\frac{1}{2}-2\pi i \frac{m}{h}}\left(\frac{r}{\lambda}\right)\right|.
	\end{split}
\end{equation}
By applying Lemma \ref{lemma::1.4} to Eq.~\eqref{eq::2.29}, one has
\begin{equation}\label{eq::EaA}
	E_A\leq\sum_{m=1}^{\infty}\frac{2\sqrt{2\pi}e^{1/6}}{r}e^{-\frac{\pi^2m}{h}}=\frac{2\sqrt{2\pi}}{r}\frac{e^{1/6-\frac{\pi^2}{h}}}{1-e^{-\frac{\pi^2}{h}}}.
\end{equation}

Since both Eqs.~\eqref{eq::1.11} and ~\eqref{eq::1.12} exhibit monotonically decreasing behavior with $M_1$ and $M_2$, and approach zero at the limit $M_1, M_2\rightarrow \infty$, one can achieve any desired accuracy $\varepsilon$ by the following procedure. First, one takes
\begin{equation}\label{eq::2.30}
	h\leq h^*=\frac{\pi^2}{\log(\varepsilon+6\sqrt{2\pi}e^{1/6}\delta^{-1})+\log\varepsilon^{-1}}
\end{equation}
so that $E_A\leq \varepsilon/3$ for all $r\in[\delta,R]$ according to Eq.~\eqref{eq::EaA}. Next, we solve for $M_{_{\scriptstyle *}}$, $M{\scriptstyle *}$ and $M^*$ as the solutions of equations
\begin{equation}\label{eq::2.31}
	\frac{2e^{\frac{t_{-M_{*}}}{2}-\frac{1}{4\lambda^2}e^{-t_{-M_{*}}}}}{\sqrt{\pi}}=\varepsilon/3,\quad \frac{1}{\delta}e^{-\delta^2e^{t_{M\scriptstyle *}}}=\varepsilon/3,\quad\text{and}\quad M^*= -1-\dfrac{1}{h}\left[2\log(\sqrt{2}R)+t_0\right],
\end{equation}
and set $M_1\geq M_*$ and $M_2\geq \max\{M*,~M^*\}$, ensuring $E_{T,1},~E_{T,2}\leq \varepsilon/3$ (according to Eqs.~\eqref{eq::1.11} and \eqref{eq::ET2}) and satisfying Eq.~\eqref{eq::t_m2}. Since $M_1$ and $M_2$ should be integers in the SOG series, we set $M_*$, $M*$ and $M^*$ as the nearest integers larger than their respective solutions from Eq.~\eqref{eq::2.31}. Finally, by applying the triangle inequality, Eq.~\eqref{eq::1.2} holds and the proof is completed.
\end{proof}

The total number of terms in the SOG approximation given by Theorem \ref{thm::1.1} is equal to $M_{\text{tot}}=M_1+M_2+1$. If we fix $\varepsilon$, $h$ and $t_0$, the solutions $M_{_{\scriptstyle *}}$, $M{\scriptstyle *}$ and $M^*$ of Eq.~\eqref{eq::2.31} depend on $\lambda$, $R$, and $\delta$ and can be determined explicitly. The corresponding criteria for $M_1$ and $M_2$ are given as
\begin{equation}\label{eq::2.33}
M_1\geq \frac{1}{h}\left[t_0-\log\left(\frac{\pi\varepsilon^2}{36}\right)-W\left(\frac{18}{\pi\lambda^2\epsilon^2}\right)\right]
\end{equation}
and
\begin{equation}\label{eq::2.34}
M_2\geq \frac{1}{h}\max\left\{-t_0-2\log\delta+\log\log(3\epsilon^{-1}\delta^{-1}), -t_0-h-2\log(\sqrt{2}R)\right\},
\end{equation}
where $W(\cdot)$ represents the Lambert $W$ function~\cite{abramowitz1964handbook}, defined implicitly as $W(x)e^{W(x)}=x$. The Lambert $W$ can be efficiently calculated using existing libraries for special functions. These criteria are near-optimal and straightforward to apply in practical calculations.

 

\subsection{SOG decomposition with high-order smoothness}\label{sec::2.3}
In this subsection, we demonstrate how to use the SOG approximation to construct a new decomposition with high-order smoothness. Unlike the work on u-series~\cite{DEShaw2020JCP,Liang2023SISC}, which is applicable only to power functions and solving for the cutoff $r_c$ to satisfy the continuity condition, we solve the continuity equation for the starting point of the trapezoidal rule for higher regularity, resulting in a highly efficient SOG decomposition for the general Yukawa kernel.

Given the screening length $\lambda$, near-field cutoff $r_c$, and tolerance $\varepsilon$, substituting $\delta = r_c$ and $R = 33\lambda$ (where the value of Yukawa kernel beyond $R$ is less than $10^{-16}$) into Theorem~\ref{thm::1.1} and using Eqs.~\eqref{eq::2.33} and \eqref{eq::2.34} yields an SOG approximation on $r\geq r_c$:
\begin{equation}\label{eq::far}
\mathcal{F}_{h}^{t_0}(r)=h\sum_{m=-M_1}^{M_2}f\left(t_m,r\right),
\end{equation}
with the error uniformly controlled by $\varepsilon$ for $r\geq r_c$. We decompose the Yukawa kernel $Y(r)$ as the sum of near- and far-field parts, 
\begin{equation}\label{eq::2.36}
	Y(r)\rightarrow\mathcal{N}_{h}^{t_0}(r)+\mathcal{F}_{h}^{t_0}(r),
\end{equation} 
where the far part is the SOG approximation itself, and the near part is given by
\begin{equation}\label{eq::short}
	\mathcal{N}_{h}^{t_0}(r)=\begin{cases}
		Y(r)-\mathcal{F}_{h}^{t_0}(r),\quad &\text{if }r< r_c,\\[2.2em]0,&\text{if }r\geq r_c.
	\end{cases}
\end{equation}
To ensure the continuity at $r_c$, and consequently, throughout the entire real axis, one chooses the starting point $t_0$ as the smallest positive root of
\begin{equation}\label{eq::2.38}
 Y(r_c)-\mathcal{F}_{h}^{t_0}(r_c)=0.
\end{equation}
It can be seen that this SOG decomposition reproduces $Y(r)$ exactly for $r<r_c$, and the truncation error for $r\geq r_c$ is expected to be small as the pointwise error of far-field SOG approximation can be uniformly bounded. 
It should be noted that such a construction relies on the existence of roots of Eq.~\eqref{eq::2.38}, and is exactly given by the following theorem.
\begin{theorem}
An SOG approximation can be obtained as described in Section~\ref{sec::SOGapproximation}, such that there are precisely two roots on $t_0\in[0,h)$ that strictly satisfies Eq.~\eqref{eq::2.38}. Precisely, these two roots are
	\begin{equation}\label{eq::roots}
		\alpha_1=Ch+\frac{h}{4}-\log(2\lambda r_c),\quad\alpha_2=Ch+\frac{3h}{4}-\log(2\lambda r_c),
	\end{equation}
    where $C$ is a constant enforcing that $\alpha_1$,$\alpha_2$ $\in[0, h)$.
\end{theorem}
\begin{proof}[Proof]
	Let $M_1$ and $M_2$ approach the infinity such that the approximation error is exactly $E_{A}$. Recalling Eq.~\eqref{eq::1.9.}, it can be observed that the expression of $E_A$ incorporates an oscillatory factor $(2\lambda r)^{2\pi im/h}e^{2\pi i t_0m/h}$. As the left-hand side of Eq.~\eqref{eq::1.7} is composed entirely of real numbers, there is no need to take into account errors in the imaginary part. Consequently, $E_A$ exhibits periodicity with respect to $t_0$ and has a period of $h$, with a phase shift of $\log(2\lambda r)$. By using the properties of cosine function over a period, one completes the proof.
\end{proof}

The smoothness of SOG decomposition can be increased by modifying the weights and bandwidths of the far-field Gaussians. For instance, a $C^1$ decomposition requires the force continuity condition
\begin{equation}\label{eq::1.25}
\frac{d}{dr}\left[Y(r)-\mathcal{F}_{h}^{t_0}(r)\right]{\bigg|}_{r=r_c}=0
\end{equation}
to be satisfied. For fixed $h$ and $t_0$, a $C^1$ decomposition can be constructed by varying the coefficient of the narrowest Gaussian of the SOG. Let us write
\begin{equation}\label{eq::2.41}
	\mathcal{F}_{h}^{t_0}(r)=h\left[w_{M_2}f\left(t_{M_2},r\right)+\sum_{m=-M_1}^{M_2-1}f\left(t_m,r\right)\right]
\end{equation} 
and express $w_{M_2}$ as a function of $t_0$ according to Eq.~\eqref{eq::2.38}, 
\begin{equation}\label{eq::2.42}
	w_{M_2}=\frac{1}{hf(t_{M_2},r_c)}\left[Y(r_c)-h\sum\limits_{m=-M_1}^{M_2-1}f\left(t_m,r_c\right)\right].
\end{equation}
Then, one solves Eq.~\eqref{eq::1.25} to determine $t_0$. It is preferable to adjust the parameters defining the narrowest Gaussian to prevent large errors far from the cutoff radius $r_c$. Since Eq.~\eqref{eq::1.25} is highly-nonlinear, it can be solved by applying the Newton's method and may have many solutions. A principle by rule of thumb is to choose a solution such that the residual of the $C^2$ continuity condition is minimized. If a $C^2$ continuity is desired, the bandwidth of the narrowest Gaussian can be also adjusted. For higher order of smoothness this process can be repeated until all conditions are satisfied. The procedure for constructing the SOG decomposition is summarized in Algorithm~\ref{al::SOG}.

\begin{algorithm}[!ht]
	\caption{(SOG decomposition of the Yukawa kernel)}\label{al::SOG}
	\begin{algorithmic}[1]
\State \textbf{Input}: The screening length $\lambda$, near-field cutoff $r_c$, error tolerance $\epsilon$, and the required order of smoothness for the decomposition. 
\State Substituting $\delta=r_c$ and $\varepsilon$ into Eq.~\eqref{eq::2.30} to compute the step size $h$.
\State  Solve the $C^0$ condition Eq.~\eqref{eq::2.38} and choose the smallest one in $[0,h)$ to obtain the starting point $t_0$. \State Substituting $\lambda$, $\varepsilon$, $h$, $t_0$, $\delta=r_c$ and $R=33\lambda$ into~\eqref{eq::2.33} and \eqref{eq::2.34} to obtain $M_1$ and $M_2$.
\State Using the resulting SOG approximation to construct the far-field and near-field components of the SOG decomposition, as described in Eqs.~\eqref{eq::far} and \eqref{eq::short}, respectively. 
\State If the force continuity is required, solve for $w_{M_2}$ by substituting Eq.~\eqref{eq::2.42} into Eq.~\eqref{eq::1.25} to ensure $C^1$ continuity in the decomposition. For higher-order smoothness, repeat this process, adjusting other parameters as needed, until all conditions are satisfied.
\State \textbf{Output}: The SOG decomposition of the Yukawa kernel.
\end{algorithmic}
\end{algorithm}

In MD simulations, a $C^1$ continuity is typically sufficient to ensure long-term stability. Table~\ref{tab::1.1} displays the values of $h$, $t_0$, $M_1$, $M_2$, and $w_{M_2}$ for a $C^1$-continuous SOG decomposition at different accuracy levels. To balance computational cost and precision, one sets the cutoff radius $r_c=5\lambda$ as suggested in the recent fast Ewald summation work~\cite{dharuman2017generalized}. From Table~\ref{tab::1.1}, it is observed that the values of $w_{M_2}$ are very close to $1$, and $17$ Gaussians are needed to achieve $13$-digit accuracy. For the commonly used $3$-digit precision in practical simulations, only $3$ Gaussians are required. A similar table for $C^2$-continuous SOG decomposition is provided in \ref{app::C2continuous}, showing comparable Gaussian requirements. These observations indicate that the SOG approximation is highly efficient. 

\renewcommand\arraystretch{1.3}
\begin{table}[H]
\caption{Parameter sets for $C^{1}$-continuous SOG decomposition. $M_{\text{tot}}:=M_1+M_2+1$ is the minimal number of Gaussians satisfying the error criteria on the region $[r_c,R]$ with $r_c=5\lambda$, $R=33\lambda$, and $\lambda=0.5773$.}
\centering
\begin{tabular}{ccccccc}
\hline
$\epsilon$ &$h$ &$ t_0 $ &$M_1$ &$M_2$ &$M_{\rm tot}$ &$\omega_{M_2}$ \\ \hline
		$10^{-3}$    &1.131155934143089 &0.162000562164036 &3  &-1 &3  &0.935842393886743 \\ \hline
		$10^{-4}$   &0.894984933518395 &0.107611873115997  &4  &-1 &4  &1.022759365476827 \\ \hline
		$10^{-5}$ 	  &0.740391708745519 &0.675290708263873  &5  &-1 &5  &0.985751471005324 \\  \hline
		$10^{-7}$  	&0.550285792019561 &0.106160774592953 &6  &0 &7 &1.002417841278074 \\  \hline
		$10^{-9}$ 	  &0.437859267899280 &0.078115088821801  &8  &1 &10  &1.000300173074238 \\  \hline
		$10^{-11}$    &0.363578174148321 &0.169191259198386 &11  &2  &14 &0.996791207311602 \\ \hline
		$10^{-13}$ 	&0.310844614243983 &0.009697549456972 &13  &3 &17 &1.000600251730658 \\  
		\hline
	\end{tabular}
\label{tab::1.1}
\end{table}

\subsection{Error estimate}
In this section, we extend the pointwise error estimate of the SOG approximation from Section~\ref{sec::SOGapproximation} to the energy and force calculations when the corresponding SOG decomposition is applied to the Yukawa potential. These accuracy measures are critical for practical MD simulations. Periodic boundary conditions are assumed throughout the remainder of this paper.

The decomposition error of potential can be written in the form of
\begin{equation}
	\Phi_{\text{err}}(\bm{r}_i)=\sum_{\bm{n}\in\mathbb{Z}^3}\sum_{j=1}^{N}q_jK(|\bm{r}_{ij}+\bm{n}\circ\bm{L}|),
\end{equation}
where $\bm{r}_{ij}:=\bm{r}_j-\bm{r}_i$, $\bm{L}=(L_x,L_y,L_z)$, and the kernel function is defined by 
\begin{equation}
	K(r):=\left[Y(r)-h\sum_{m=-M_1}^{M_2}f\left(t_m,r\right)\right]H(r-r_c).
\end{equation}
Here, $H(r)$ denotes the Heaviside step function with $H(r)=1$ for $r\geq0$ and $0$ otherwise. Let us define the 3D Fourier transform in conjunction with its conjugated inverse transform as
\begin{equation}
	\widehat{f}(\boldsymbol{k}):=\int_{\Omega} f(\boldsymbol{r}) e^{-i \boldsymbol{k} \cdot \boldsymbol{r}} d \boldsymbol{r} \quad \text { and } \quad f(\boldsymbol{r})=\frac{1}{V} \sum_{\boldsymbol{k}} \widehat{f}(\boldsymbol{k}) e^{i \boldsymbol{k} \cdot \boldsymbol{r}},
\end{equation}
where $V=L_xL_yL_z$ represents the volume of simulation box and $\bm{k}$ is the Fourier mode defined by
\begin{equation}
	\bm{k}=2\pi\left(m_x/L_x,m_y/L_y,m_z/L_z\right)
\end{equation}
with $m_x,m_y,m_z\in\mathbb{Z}$. By the Fourier transform, the truncation errors in energy and force can be analytically written as
\begin{equation}\label{eq::Uerr}
	U_{\text{err}}=\frac{1}{2V}\sum_{\bm{k}}|\rho(\bm{k})|^2\widehat{K}(k)\quad\text{and}\quad \bm{F}_{\text{err}}(\bm{r}_i)=\frac{q_i}{V}\sum_{\bm{k}}\bm{k}|\text{Im}(e^{-i\bm{k}\cdot\bm{r}_i}\rho(\bm{k}))|\widehat{K}(k),
\end{equation}
where $\widehat{K}(k)$ is the Fourier transform of $K(r)$. Starting from Eq.~\eqref{eq::Uerr}, one can establish Theorem~\ref{thm::2.8}. Firstly, one needs Lemma~\ref{lemma::fourier} for the Fourier transform of a radially symmetric function~\cite{stein2011fourier}.
\begin{lemma}\label{lemma::fourier}
	Assume that the Fourier transform of $g(\bm{r})$ exists. If $g(\bm{r})$ is a radially symmetric function in 3D, its Fourier transform is also radially symmetric, expressed by
	\begin{equation}
		\widehat{g}(k)=4 \pi \int_0^{\infty} \frac{\sin (k r)}{k} g(r) r d r.
	\end{equation}
\end{lemma} 

\begin{theorem}\label{thm::2.8}
The truncation errors of energy and force can be estimated by
	\begin{equation}\label{eq::2.28}
		\begin{split}
		|U_{\emph{err}}|&\simeq O\left(e^{-\pi^2/h}+e^{-(M_1+1)h/2-(4\lambda^2)^{-1}e^{(M_1+1)h}}+e^{-r_c^2e^{t_{M_2+1}}}\right),\\
		|\bm{F}_{\emph{err}}(\bm{r}_i)|&\simeq O\left(e^{-\pi^2/h}+e^{-3(M_1+1)h/2-(4\lambda^2)^{-1}e^{(M_1+1)h}}+e^{-r_c^2e^{t_{M_2+1}}+t_{M_2+1}}\right),
		\end{split}
	\end{equation}
    where $\simeq$ indicates ``asymptotically equal'' as $h\rightarrow 0$. 
\end{theorem}

\begin{proof}
Similar to the proof of Theorem~\ref{thm::1.1}, $K(r)$ can be divided into three parts, 
\begin{equation}
\begin{split}
K(r)&=h\left[\sum_{m=-\infty}^{-M_1-1}f(t_m,r)+\sum_{m=M_2+1}^{\infty}f(t_m,r)+\left(\frac{e^{-r/\lambda}}{r}-h\sum_{m=-\infty}^{\infty}f(t_m,r)\right)\right]H(r-r_c)\\
&:=K_{T,1}(r)+K_{T,2}(r)+K_A(r),
\end{split}
\end{equation}
where the first two terms arise from truncating the Gaussian terms, and the third term is due to the quadrature error of the trapezoidal rule. Therefore, $U_{\text{err}}$ can be rewritten as
\begin{equation}\label{eq::split}
	\begin{split}
	U_{\text{err}}=\frac{1}{2V}\sum_{\bm{k}}|\rho(\bm{k})|^2\left[\widehat{K}_{T,1}(k)+\widehat{K}_{T,2}(k)+\widehat{K}_A(k)\right]:=U^{T,1}_{\text{err}}+U^{T,2}_{\text{err}}+U_{\text{err}}^A.
	\end{split}
\end{equation} 
For each component, the sum over Fourier modes $\bm{k}$ can be asymptotically and safely approximate by an integral~\cite{kolafa1992cutoff,liang2023error},
\begin{equation}\label{eq::asy}
	\sum_{\bm{k}} \simeq \frac{V}{(2 \pi)^3} \int_{0}^{\infty} k^2 d k \int_{-1}^1 d \cos \theta \int_0^{2 \pi} d \varphi,
\end{equation}
where $(k,\theta,\varphi)$ are the spherical coordinates and $\simeq$ indicates asymptotically equal in the mean field limit. By applying Eq.~\eqref{eq::asy} to Eq.~\eqref{eq::split} and using Lemma~\ref{lemma::fourier}, one can derive three oscillating Fourier integrals which can be estimated by directly extending the results of $E_{T,1}$, $E_{T,2}$, and $E_{A}$ in the proof of Theorem~\ref{thm::1.1}. Similarly, this procedure can be employed to estimate the error of force. 
\end{proof}

\begin{figure}[!ht]
	\centering
\includegraphics[width=0.9\linewidth]{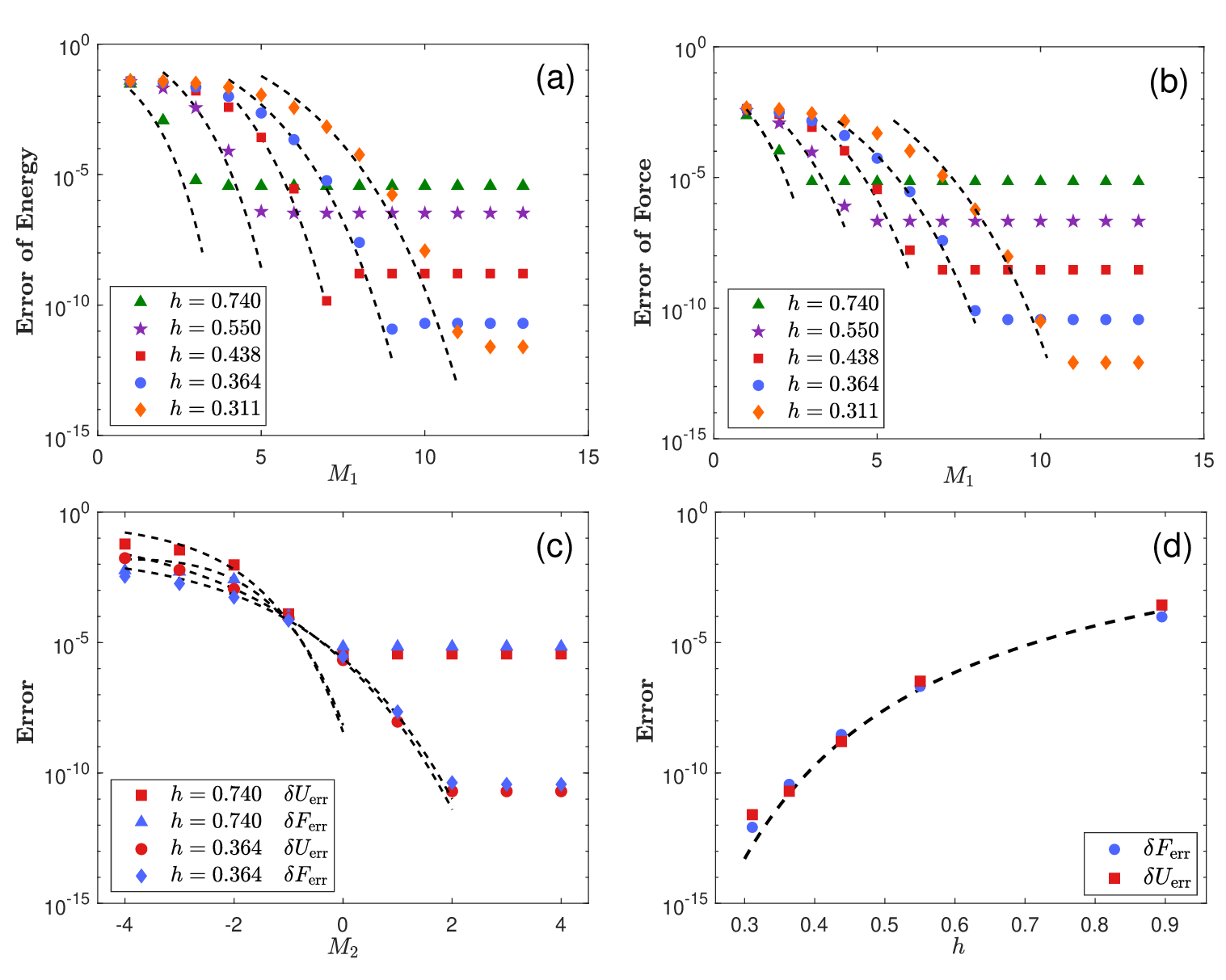}
	\caption{Truncation errors of the SOG decomposition for the YOCP system with $\lambda=0.5773$. (a-b): Relative errors of the energy and the force, respectively, as functions of $M_1$, with other parameters set to not impact the accuracy. (c): Relative errors for the energy and the force against $M_2$ for $h=0.740$ and $0.364$. (d): Relative errors as functions of the discretization size $h$. The dashed lines in all panels represent the corresponding theoretical estimates by Theorem~\ref{thm::2.8}.}
	\label{fig:truncationerror}
\end{figure}

\begin{figure}[!ht]
	\centering
\includegraphics[width=0.9\linewidth]{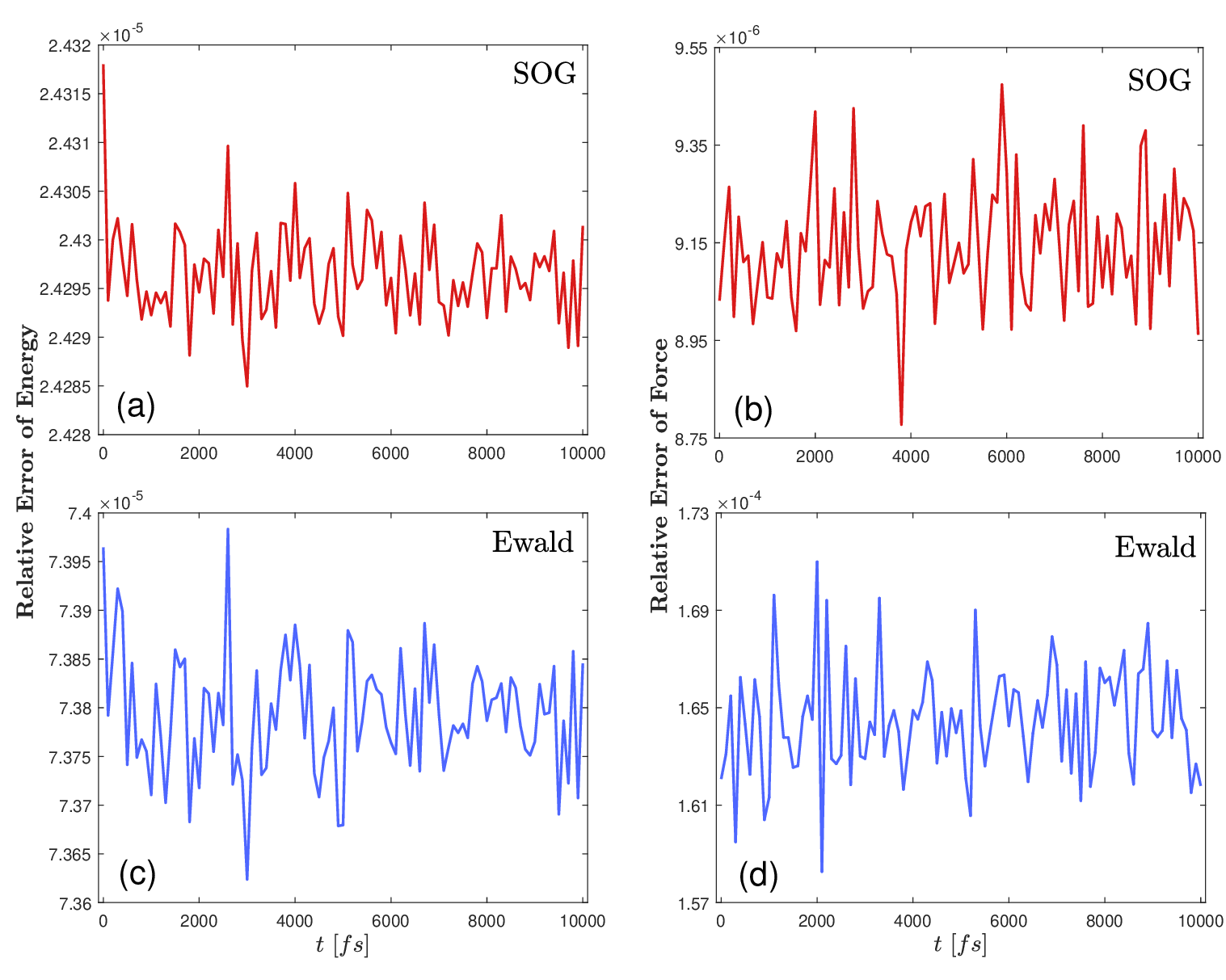}
	\caption{Relative errors in energy and force for different kernel decompositions across 100 configurations of a YOCP system. The configurations were sampled at 100 $fs$ intervals from a 10000 $fs$ simulation.  (a-b): Results using the $C^1$-continuous SOG decomposition at an accuracy of $10^{-5}$ (third line in Table~\ref{tab::1.1}). (c-d): Results using the Ewald decomposition, with a matching Fourier space decay rate.}
	\label{fig:sogrelative}
\end{figure}

Theorem~\ref{thm::2.8} indicates that the decay rate with the number of Gaussians is $O(e^{-t_{M_1+1}/2})$ for the energy and $O(e^{-3t_{M_1+1}/2})$ for the force, respectively, with lower bounds set by the other two terms in Eq.~\eqref{eq::2.28}. To validate Theorem~\ref{thm::2.8}, numerical calculation on a Yukawa one-component plasma (YOCP) system with side length $1nm$ and $10^5$ monovalent cations is performed, averaging over $100$ different configurations. Figure~\ref{fig:truncationerror}(a-b) shows relative errors in the energy and the force versus $M_1$, with dashed lines representing $O(e^{-t_{M_1+1}/2})$ and $O(e^{-3t_{M_1+1}/2})$ scaling, respectively. Figure~\ref{fig:truncationerror}(c) shows errors as functions of $M_2$, with dashed lines indicating $O(e^{-r_c^2e^{t_{M_2+1}}})$ and $O\left(e^{-r_c^2e^{t_{M_2+1}}+t_{M_2+1}}\right)$ scaling for the energy and the force. Figure~\ref{fig:truncationerror}(d) illustrates the errors versus $h$, with dashed lines show $O\left(e^{-\pi^2/h}\right)$ scaling. All the results agree well with Theorem~\ref{thm::2.8}. 

Additionally, we compare our results with those of the Ewald decomposition~\cite{dharuman2017generalized}, adjusting the parameters to match the decay rate for the Fourier space of its far-field component to that of the $C^1$-continuous SOG decomposition at an accuracy of $10^{-5}$. For both methods, the near-field cutoff is set to $r_c = 2.8865~\mathring{A}$. The average relative errors in energy and force are shown in Figure~\ref{fig:sogrelative}. The results indicate that the errors for the Ewald decomposition are approximately $3$ and $10$ times larger, respectively, than those of the SOG decomposition proposed in this paper. This improvement, attributed to the $C^1$ smoothness of the SOG decomposition, is expected to benefit fast algorithms.

\section{Fast algorithm}\label{sec::fast}
In this section, we first derive the Fourier spectral expansion for the far-field energy and force using the SOG decomposition introduced in the previous section. We then review the random mini-batch technique and explain why existing sampling strategies are inefficient for Yukawa systems. Finally, we propose a novel adaptive importance sampling strategy tailored for the Yukawa kernel, leading to a fast and adaptive RBSOG method with $O(N)$ complexity.

\subsection{Fourier spectral expansion}
By using the SOG decomposition in Eq.~\eqref{eq::2.36}, the energy of Yukawa systems can be expressed as a combination of near-field, far-field, and self-energy terms:
\begin{equation}\label{eq::3.1}
	U=U_{\mathcal{N}}+U_{\mathcal{F}}+U_{\text{self}}.
\end{equation}
Here, the first two parts are given by
\begin{equation}
	U_{\mathcal{N}}=\frac{1}{2} \sum_{\boldsymbol{n}}\,^{\prime} \sum_{i, j} q_i q_j \mathcal{N}_{h}^{t_0}\left(\left|\boldsymbol{r}_{i j}+\boldsymbol{n} \circ \boldsymbol{L}\right|\right),
\end{equation}
and
\begin{equation}
	U_{\mathcal{F}}=\frac{1}{2}\sum_{\boldsymbol{n}} \sum_{i, j} q_i q_j \mathcal{F}_{h}^{t_0}\left(\left|\boldsymbol{r}_{i j}+\boldsymbol{n} \circ \boldsymbol{L}\right|\right),
\end{equation}
respectively, where the prime indicates that the case $i=j$ with $\bm{n}=\bm{0}$ is excluded in the summation, and ``$\circ$'' represents the Hadamard product of two vectors. We adopt the $C^1$-continuous SOG decomposition, which is typically sufficient for stable simulations. The self-energy term $U_{\text{self}}$ is to exclude unwanted self interactions, and is given by
\begin{equation}
	U_{\text{self}}=-\frac{1}{2}\sum_{i=1}^{N}q_{i}^2\mathcal{F}_{h}^{t_0}\left(0\right)=-\frac{h}{2}\sum_{i=1}^{N}q_{i}^2\left[w_{M_2}f\left(t_{M_2},0\right)+\sum_{m=-M_1}^{M_2-1}f\left(t_m,0\right)\right].
\end{equation}

By the SOG decomposition, the sum in $U_{\mathcal{N}}$ exhibits rapid convergence, allowing for truncation at a real-space cutoff $r_c$. In contrast, the sum in $U_{\mathcal{F}}$ converges slowly but the kernel function is smooth, and can thus be treated in Fourier space for efficiency. By the Fourier transform, the far-field part is rewritten as
\begin{equation}\label{eq::3.5}
	\begin{split}
		U_{\mathcal{F}}=\frac{1}{2V}\sum_{\bm{k}}\widehat{\mathcal{F}}_{h}^{t_0}(\bm{k})\left|\rho(\bm{k})\right|^2, 
	\end{split}
\end{equation} 
where 
\begin{equation}
	\widehat{\mathcal{F}}_{h}^{t_0}(\bm{k})=\pi^{3/2}h\sum_{m=-M_1}^{M_2}{}^\#f\left(t_m,0\right)e^{-3t_{m}/2-e^{-t_{m}}k^2/4}
\end{equation}
represents the Fourier transform of $\mathcal{F}_{h}^{t_0}(\bm{r})$, and the symbol ``$\#$'' represents the multiplication of the $m=M_2$ term by the factor $w_{M_2}$. The reciprocal sum over $\bm{k}$ in Eq.~\eqref{eq::3.5} is absolutely convergent due to the fast decay of Gaussian functions.

With the Fourier spectral expansion of the far-field energy, the force exerts on the $i$th particle is derived by taking the negative gradient of the energy, expressed by
\begin{equation}\label{eq::3.8}
	\begin{split}
		\bm{F}(\bm{r}_i)&=\sum_{j\in\mathcal{I}_i} q_i q_j\left[F_{\mathcal{N},1}(r_{ij})+F_{\mathcal{N},2}(r_{ij})\right]\bm{r}_{ij}-\frac{q_{i}}{V}\sum_{\bm{k}}\bm{k}\widehat{\mathcal{F}}_{h}^{t_0}(\bm{k})\operatorname{Im}\left(e^{-i \boldsymbol{k} \cdot \boldsymbol{r}_i} \rho(\boldsymbol{k})\right)\\[0.7em]
		&:=\bm{F}_{\mathcal{N}}(\bm{r}_i)+\bm{F}_{\mathcal{F}}(\bm{r}_i)
	\end{split}
\end{equation}
where $\mathcal{I}_i$ denotes the neighbor list of the $i$th particle, and $F_{\mathcal{N},1}$ and $F_{\mathcal{N},2}$ are defined by 
\begin{equation}
	F_{\mathcal{N},1}(r):=\frac{e^{-r/\lambda}}{r^3}(1+r/\lambda),
\end{equation}
\begin{equation}
	F_{\mathcal{N},2}(r):=-2h\left[w_{M_2}e^{t_{M_2}}f\left(t_{M_2},r\right)+\sum_{m=-M_1}^{M_2-1}e^{t_m}f\left(t_m,r\right)\right].
\end{equation}
In Eq.~\eqref{eq::3.8}, $\bm{F}_{\mathcal{N}}$ and $\bm{F}_{\mathcal{F}}$ are the near- and far-field parts of the force, and are treated in the real and Fourier spaces, respectively. 

Let us study the complexity of evaluating $U$ and $\bm{F}$ using Eqs.~\eqref{eq::3.1} and ~\eqref{eq::3.8}. The computational cost for the near-field part is $O( r_c^3N)$ due to the truncation. 
The cost of evaluating the far-field part depends on the specific method. For direct truncation, the Fourier space cutoff satisfies $k_c\sim e^{t_{M_2}/2}:=C_{M_2}$, which is inversely proportional to the bandwidth of the narrowest Gaussian. By substituting $\delta=r_c$ into Eq.~\eqref{eq::2.34}, it can be deduced that $C_{M_2}\sim r_c$. Since the Fourier modes within the cutoff are proportional to $k_c^3=O(1/r_c^3)$, minimizing the total cost results in a complexity of $O(C_{M_2}^{3/2})$. Alternatively, if the FFT is used for acceleration, the mesh spacing is also proportional to $C_{M_2}$~\cite{shan2005gaussian}, leading to a cost that scales as $O(\mathcal{G}\log\mathcal{G}+N\log N)$ with $\mathcal{G}\sim1/C_{M_2}^3$ the number of grids. The communication cost of FFT is $O(N)$~\cite{ayala2021scalability} for each calculation.

\subsection{Random batch importance sampling technique}\label{subsec::RBS}
Another way to evaluate $U_{\mathcal{F}}$ and $\bm{F}_{\mathcal{F}}$ is the so-called random batch importance sampling methods, including the RBE~\cite{Jin2020SISC} and RBSOG~\cite{Liang2023SISC}, which were originally developed for the pure Coulomb kernel ($\lambda\rightarrow\infty$). The key difference between the RBE and RBSOG lies in the decomposition method: the RBE is based on the Ewald decomposition, while the RBSOG utilizes the SOG decomposition. In this work, we aim to extend the RBSOG framework to efficiently handle Yukawa systems. 

Unlike deterministic methods that rely on FFT or FMM-based techniques to reduce complexity, the RBSOG method utilizes mini-batch stochastic approximation over Fourier modes, combined with importance sampling to achieve a reduced variance. Let us consider the Fourier sum over $\bm{k}$ for $\bm{F}_{\mathcal{F}}$. One can alternatively understand the Fourier sum as an expectation 
\begin{equation}
\bm{F}_{\mathcal{F}}^{\mathscr{P}}(\bm{r}_i)=\frac{q_i}{V}\mathbb{E}_{\bm{k}\sim \mathscr{P}(\bm{k})}\left[\frac{\bm{k}\widehat{\mathcal{F}}_{h}^{t_0}(\bm{k})\operatorname{Im}\left(e^{-i \boldsymbol{k} \cdot \boldsymbol{r}_i} \rho(\boldsymbol{k})\right)}{\mathscr{P}(\bm{k})}\right],
\end{equation}
where $\mathscr{P}(\bm{k})$ is a discrete measure, referred to as the ``importance''. Instead of computing the summation directly or using the FFT, a mini-batch of $\bm{k}$ (with batch size $P$) sampled from $\mathscr{P}(\bm{k})$ are employed to estimate the expectation, resulting in an efficient stochastic method.
 
However, when designing algorithms for Yukawa systems, directly applying the random mini-batch idea presents significant challenges. To better understand this issue, we review some importance sampling strategies previously proposed for the Coulomb case and discuss why they are inefficient for Yukawa systems. In \cite{Jin2020SISC}, it is suggested that $\mathscr{P}(\bm{k})$ should be chosen as the far-field Gaussian itself, since the Gaussian is summable and can be normalized into a discrete distribution. In~\cite{Liang2023SISC}, it is proposed that $\mathscr{P}(\bm{k})$ should be chosen as $k^2$ multiplied by the SOG $\widehat{\mathscr{F}}_{h}^{t_0}$, which improves efficiency compared to the approach in \cite{Jin2020SISC}. However, these strategies are effective only under the limit $\lambda\rightarrow\infty$ and the charge neutrality condition $\sum_{i=1}^{N}q_i=0$. In Yukawa systems, where electrons are implicitly treated and $\lambda$ varies within $(0,\infty)$, both conditions are typically \emph{not satisfied}. As a result, these choices of $\mathscr{P}(\bm{k})$ can lead to significant variance, particularly in weakly coupled YOCPs, as will be numerically demonstrated in Section~\ref{subsec::variancereduction}.

\subsection{Adaptive importance sampling under the random batch framework}
\label{subsec:important}
We propose a new adaptive importance sampling strategy to achieve a near-optimal variance reduction in Fourier-space calculations, resulting in a random batch sum-of-Gaussians method for the Yukawa kernel. This method approximates $\bm{U}_{\mathcal{F}}$ and $\bm{F}_{\mathcal{F}}$ with $O(N)$ computational cost and $O(1)$ communication cost per calculation.  

More precisely, our purpose is to find a proper sampling measure in the form
\begin{equation}\label{eq::3.124}
	\mathscr{P}(\bm{k}):=\dfrac{\mathscr{F}(\bm{k})\widehat{\mathcal{F}}_{h}^{t_0}(\bm{k})}{S},
\end{equation}
where the SOG $\widehat{\mathcal{F}}_{h}^{t_0}$ is added to make $\mathscr{P}(\bm{k})$ summable, $S$ is the normalization factor, and $\mathscr{F}(\bm{k})$ serves as a correction. By Eq.~\eqref{eq::3.124}, one has random estimators
\begin{equation}
\label{eq::UF}
U_{\mathcal{F}}^{\mathscr{P},*}=\frac{S}{2PV}\sum_{\ell=1}^{P}\frac{\left|\rho(\bm{k}_{\ell})\right|^2}{\mathscr{F}(\bm{k}_{\ell})}\quad\text{and}\quad \bm{F}_{\mathcal{F}}^{\mathscr{P},*}(\bm{r}_i)=-\frac{S}{2PV}\sum_{\ell=1}^{P}\frac{\nabla_{\bm{r}_i}\left|\rho(\bm{k}_{\ell})\right|^2}{\mathscr{F}(\bm{k}_{\ell})},
\end{equation}
to approximate the far-field energy and force, respectively. Formally, an ``ideal'' choice for $\mathscr{F}(\bm{k})$ is $\mathscr{F}(\bm{k})=\left|\rho(\bm{k})\right|^2$ so that the variances of $U_{\mathcal{F}}^{\mathscr{P},*}$ and $\bm{F}_{\mathcal{F}}^{\mathscr{P},*}$ are minimized. However, this choice is not applicable as the knowledge of $\rho(\bm{k})$ is unknown. Instead, one takes 
\begin{equation}\label{eq::3.15}
	\mathscr{F}(\bm{k})=\langle\left|\rho(\bm{k})\right|^2\rangle,
\end{equation}
where $\langle\cdot\rangle$ represents the ensemble average. It serves as a good approximation of $|\rho(\bm{k})|$, as the ensemble average reflects the long-term behavior and captures the attributes of all accessible configurations~\cite{Frenkel2001Understanding}, while being computationally cheaper. In practice, we use a cost-effective formula to approximate $\langle\left|\rho(\bm{k})\right|\rangle$. In the liquid theory~\cite{hansen2013theory}, $\mathscr{F}(\bm{k})$ in Eq.~\eqref{eq::3.15} is known as the charge structure factor reflecting the influence of different kernel. Under the well-known Debye-H$\ddot{\text{u}}$ckel theory~\cite{hansen2013theory,hu2022symmetry}, the linearly screened system described in Section~\ref{Sec::2.1} leads to the charge response as $k\rightarrow 0$:
\begin{equation}\label{eq3::3.22}
	\langle\left|\rho(\bm{k})\right|^2\rangle=\dfrac{k_{\text{B}}TV}{k_{\text{B}}TV/\sum_{i=1}^{N}q_{i}^2+\epsilon(\bm{k})^{-1}/k^{2}}+O(k^4),
\end{equation}
where $k_{\text{B}}$ is the Boltzmann constant and $T$ is the temperature. The first term in the right hand side of Eq.~\eqref{eq3::3.22} is adaptive for different dielectric response functions $\epsilon(\bm{k})$. In practice, we truncate the right-hand side of Eq.~\eqref{eq3::3.22} to its first term and consider it an efficient correction to the sampling measure. During NVT ensemble simulations, the instantaneous temperature is computed on-the-fly, while other terms can be precomputed. If the simulation includes cell volume fluctuations, particle birth/death processes, or the variation of $\epsilon(\bm{k})$, the corresponding terms are also calculated in real time to adaptively obtain an approperiate correction of $\langle\left|\rho(\bm{k})\right|^2\rangle$ to enhance the sampling.

While the approximation in Eq.~\eqref{eq3::3.22} is less accurate for large-$k$ modes, it remains efficient due to the following reasons. In our RBSOG method, the Fourier transform is applied exclusively to the long-range component of the SOG decomposition, making long-wave modes more important as they correspond to spatially slow variations. These merits should also be inherited by the random batch importance sampling, and is optimal in sense of average energy fluctuation which will be demonstrated in Theorem~\ref{thm::3.2}. Interestingly, for the pure Coulomb kernel, i.e. $\lambda\rightarrow\infty$ and $\epsilon(\bm{k})\equiv 1$, one has $\langle\left|\rho(\bm{k})\right|^2\rangle= k_{\text{B}}TVk^2+O(k^4)$ by Eq.~\eqref{eq3::3.22},
which is also consistent with the theoretical work of
Stillinger and Lovett~\cite{stillinger1968general}. The detailed procedure for MD simulations is summarized in Algorithm~\ref{al::RBSOG}.

\begin{algorithm}[!ht]
	\caption{(Random batch sum-of-Gaussians algorithm for Yukawa systems)}\label{al::RBSOG}
	\begin{algorithmic}[1]
\State  \textbf{Input}: The screen length $\lambda$, real-space cutoff $r_c$, step size $\Delta t$ for time integration, total time steps $N_{T}$, batch size $P$. Initialize the size of the simulation box $\bm{L}=(L_x,L_y,L_z)$, as well as the positions, velocities, and charges of all particles. Construct the SOG decomposition by using Algorithm~\ref{al::SOG}. 
\For {$n \text{ in } 1: N_{T}$}
\State Draw $P$ Fourier modes $\{\bm{k}_{\ell}\}_{\ell=1}^{P}$ from $\mathscr{P}(\bm{k})$ and then approximate the far-field force using the estimator $\bm{F}_{\mathcal{F}}^{\mathscr{P},*}$, where the adaptive importance sampling strategy (say Eq.~\eqref{eq3::3.22}) is applied to reduce the variance.
\State Compute the near-field force $\bm{F}_{\mathcal{N}}$ by Eq.~\eqref{eq::3.8}.
\State If desired, the energy is obtained via Eq.~\eqref{eq::3.1}, while its far-field part is approximated by Eq.~\eqref{eq::UF}.
\State Integrate Newton's equations for $\Delta t$ time with appropriate integration scheme and thermostat.
\EndFor
\State \textbf{Output}: The configurations at each time step during the simulations.
\end{algorithmic}
\end{algorithm}

\begin{remark}
Compared to FFT-based Ewald summation, the RBSOG method offers three key advantages: 1) it uses an SOG decomposition with high regularity, addressing the discontinuity issue and reducing truncation error; 2) it is mesh-free and employs random batch sampling in Fourier space calculations, avoiding the communication-intensive FFT framework; 3) the proposed variance reduction strategy is kernel-adaptive, making the algorithm easily extendable to other dielectric response functions used in plasma simulations.  
\end{remark}

In our implementation of the proposed RBSOG method, we optimize for CPU parallelization and vectorization. The near-field calculations involve the kernel precomputation and tabulation. For instance, an inner cutoff $r_{\text{in}}<r_c$ is introduced. If $0<r\leq r_{\text{in}}$, $F_{\mathcal{N},1}$ is directly computed while $F_{\mathcal{N},2}$ is  approximated via Taylor expansions. When $r_{\text{in}}<r\leq r_c$, bitmask-based table lookup~\cite{wolff1999tabulated} technique is adopted to tabulate $F_{\mathcal{N},1}+F_{\mathcal{N},2}$, following by a linear interpolation to approximate the data between successive points in the table. The above strategies make the computational cost independent of the number of Gaussians. For far-field calculations, one uses the Metropolis-Hastings method~\cite{metropolis1953equation} to sample required batches from $\mathscr{P}(\bm{k})$. After that, both these samples and particles are packaged into vectors, and $\rho(\bm{k})$ at each core is computed locally followed by applying a global reduction to reduce $\rho(\bm{k})$ and broadcasting the result back to all cores. To improve the performance, modern domain decomposition techniques~\cite{thompson2022lammps} are also integrated in our code. 

\subsection{Analysis of the RBSOG algorithm}\label{sec::3.3}
Denote the fluctuation of approximation for the Fourier part of energy and force execting on particle $i$ by
\begin{equation}\label{eq::3.19}
	\Xi_{U}=U_{\mathcal{F}}^{\mathscr{P},*}-U_{\mathcal{F}}\,\quad\text{and}\quad\,\bm{\Xi}_{\bm{F},i}=\bm{F}_{\mathcal{F}}^{\mathscr{P},*}(\bm{r}_i)-\bm{F}_{\mathcal{F}}(\bm{r}_i),
\end{equation}
respectively. The expectations and variances can be obtained by direct calculations, and are given by the following Lemma~\ref{thm::unbia}. 
\begin{lemma}\label{thm::unbia}
	The fluctuation in energy $\Xi_{U}$ and force $\bm{\Xi}_{\bm{F},i}$ have zero expectations,
	\begin{equation}
		\mathbb{E}\Xi_{U}=0,\quad\,\mathbb{E}\bm{\Xi}_{\bm{F},i}=\bm{0},
	\end{equation}
    and the variances are given by
    \begin{equation}\label{eq::3.22}
    	\mathbb{E}|\Xi_{U}|^2=\dfrac{1}{P}\left[\frac{S}{4V^2}\sum_{\bm{k}}\dfrac{\widehat{\mathcal{F}}_{h}^{t_0}(\bm{k})\left| \rho(\boldsymbol{k})\right|^4}{\langle\left|\rho(\bm{k})\right|^2\rangle/V}-\left|U_{\mathcal{F}}\right|^2\right]
    \end{equation}
    and
    \begin{equation}\label{eq::3.24}
    	\mathbb{E}|\bm{\Xi}_{\bm{F},i}|^2=\dfrac{1}{P}\left[\frac{q_i^2S}{V^2}\sum_{\bm{k}}\dfrac{k^2\widehat{\mathcal{F}}_{h}^{t_0}(\bm{k})\left|\operatorname{Im}\left(e^{-i \boldsymbol{k} \cdot \boldsymbol{r}_i} \rho(\boldsymbol{k})\right)\right|^2}{\langle\left|\rho(\bm{k})\right|^2\rangle/V}-\left|\bm{F}_{\mathcal{F}}(\bm{r}_i)\right|^2\right],
    \end{equation}
    respectively.
\end{lemma}
\noindent Lemma~\ref{thm::unbia} guarantees the consistency of stochastic approximations, i.e. $\mathbb{E}U_{\mathcal{F}}^{\mathscr{P},*}=U_{\mathcal{F}}$ and $\mathbb{E}\bm{F}_{\mathcal{F}}^{\mathscr{P},*}(\bm{r}_i)=\bm{F}_{\mathcal{F}}(\bm{r}_i)$. Eqs.~\eqref{eq::3.22}-\eqref{eq::3.24} illustrate that the variances of both energy and force scale as $O(1/P)$. More precisely, one has Theorem~\ref{thm::3.2} which holds under the mean-field assumption~\cite{hansen2013theory}.
\begin{theorem}\label{thm::3.2}
Let $\rho_r=N/V$ be the number density.~Under the mean-field assumption that particles are uniformly distributed without correlation, the adaptive importance sampling strategy defined by Eqs.~\eqref{eq::3.124}-\eqref{eq3::3.22} is ``optimal'' in sense of $\mathbb{E}|\Xi_{U}|^2\sim0$. Furthermore, the variance of force $\mathbb{E}|\bm{\Xi}_{i}|^2$ scales as $O(1/P)$, and is independent of both $N$ and the truncated number of Gaussians $M$. 

\end{theorem}
\begin{proof}
	By the definition of the structure factor, one has 
	\begin{equation}\label{eq::3.20}
		\left|\rho(\bm{k})\right|^2=\sum_{i=1}^{N}q_i^2+\sum_{\substack{i,j=1\\i\neq j}}^{N}q_iq_je^{i\bm{k}\cdot\bm{r}_{ij}}.
	\end{equation}
	The second term on the right-hand side vanishes for all $\bm{k}\neq\bm{0}$ under the mean-field assumption. Note that this assumption is widely used for error estimate of the Ewald summation~\cite{kolafa1992cutoff,deserno1998mesh}. Without loss of generality, we assume that $q_i\equiv q$ for all $i$, then the remainder term in Eq.~\eqref{eq::3.20} gives us $\left|\rho(\bm{k})\right|^2\equiv N^2q^2$ as $\bm{k}=\bm{0}$ and
	\begin{equation}\label{eq::rhok}
		\left|\rho(\bm{k})\right|^2\sim \left<\left|\rho(\bm{k})\right|^2\right>=Nq^2\quad \text{as}\quad \bm{k}\neq\bm{0}.
	\end{equation}
 Substituting this result and the normalization factor $S=V^{-1}\sum_{\bm{k}}\langle\left|\rho(\bm{k})\right|^2\rangle\widehat{\mathcal{F}}_{h}^{t_0}(\bm{k})$	into Eq.~\eqref{eq::3.22} yields
\begin{equation}\label{eq::3.27}
		\begin{split}
			\mathbb{E}|\Xi_{U}|^2&\sim \dfrac{1}{4PV^2}\left[\sum_{\bm{k}^*}\widehat{\mathcal{F}}_{h}^{t_0}(\bm{k}^*)\langle\left|\rho(\bm{k}^*)\right|^2\rangle\sum_{\bm{k}}\frac{\widehat{\mathcal{F}}_{h}^{t_0}(\bm{k})\langle\left| \rho(\boldsymbol{k})\right|^4\rangle}{\langle\left|\rho(\bm{k})\right|^2\rangle}-\left|\sum_{\bm{k}}\widehat{\mathcal{F}}_{h}^{t_0}(\bm{k})\langle\left| \rho(\boldsymbol{k})\right|^2\rangle\right|^2\right]=0.
		\end{split}
	\end{equation}
	
Next we consider the variance of force. By the mean-field theory, the term of the structure factor in Eq.~\eqref{eq::3.24} is bounded by a constant $C$ as $\bm{k}\neq\bm{0}$~\cite{Jin2020SISC},
\begin{equation}\label{eq::3.29}
	\left|\operatorname{Im}\left(e^{-i \boldsymbol{k} \cdot \boldsymbol{r}_i} \rho(\boldsymbol{k})\right)\right|^2\leq C.
\end{equation}
By this inequality and Eq.~\eqref{eq::rhok}, one has 
\begin{equation}
	\begin{split}
	\mathbb{E}|\bm{\Xi}_{F,i}|^2&\leq \dfrac{1}{P}\frac{Cq^2}{V^2}\sum_{\bm{k}^*}\widehat{\mathcal{F}}_{h}^{t_0}(\bm{k}^*)\langle\left|\rho(\bm{k}^*)\right|^2\rangle\sum_{\bm{k}}\dfrac{k^2\widehat{\mathcal{F}}_{h}^{t_0}(\bm{k})}{\langle\left|\rho(\bm{k})\right|^2\rangle}\\
	&= \dfrac{1}{P}\frac{Cq^2}{V^2}\sum_{\bm{k}^*}\widehat{\mathcal{F}}_{h}^{t_0}(\bm{k}^*)\sum_{\bm{k}}k^2\widehat{\mathcal{F}}_{h}^{t_0}(\bm{k}).
	\end{split}
\end{equation}	
One employs the integral approximation as per Eq.~\eqref{eq::asy}:
\begin{equation}
	\begin{split}
		\sum_{\bm{k}}\widehat{\mathcal{F}}_{h}^{t_0}(\bm{k})&\simeq \frac{Vh}{2\pi}\sum_{m=-M_1}^{M_2}e^{-t_m}\int_{0}^{\infty} k^2e^{-e^{-t_m}k^2/4}dk=\frac{Vh}{\sqrt{\pi}}\sum_{m=-M_1}^{M_2}e^{t_m/2}.
	\end{split}
\end{equation}
Similarly, one has 
\begin{equation}
	\sum_{\bm{k}}k^2\widehat{\mathcal{F}}_{h}^{t_0}(\bm{k})\simeq \frac{6Vh}{\sqrt{\pi}}\sum_{m=-M_1}^{M_2}e^{3t_m/2}.
\end{equation}
Since both $\sum\limits_{m=-M_1}^{M_2}e^{t_m/2}$ and $\sum\limits_{m=-M_1}^{M_2}e^{3t_m/2}$ are bounded by a constant $C_1$, one obtains
\begin{equation}
	\mathbb{E}|\bm{\Xi}_{F,i}|^2\leq \frac{CC_1h^2}{\pi P}=O\left(\frac{1}{P}\right)
\end{equation}
which is independent of $N$ and $M$. 

\end{proof}

Theorem~\ref{thm::3.2} suggests that the adaptive importance sampling strategy is near-optimal in context of the mean-field assumption for arbitrary kernel parameters. 
Since the distribution could not be strictly isotropic at each step, one only expects this strategy to maintain a relatively optimal level of variance in long-term simulations. In practice, Eq.~\eqref{eq3::3.22} is used to approximate $\langle|\rho(\bm{k})|^2\rangle$. In this case, one can also demonstrate $\mathbb{E}|\bm{\Xi}_{F,i}|^2\sim O(1/P)$, and the proof is provided in \ref{sec::energy}.

By Eqs.~\eqref{eq::3.22} and \eqref{eq::3.24}, the deviations in the random approximation are of $O(1)$, meaning that it cannot guarantee any digits of accuracy for each step. At first glance, this seems unacceptable. However, the rationale behind random batch-type methods is that as the MD evolution progresses, the random approximations accumulate over time, and Lemma~\ref{thm::unbia} and Theorem~\ref{thm::3.2} show that the averaged effect is correct. Hence, our method works due to this time-averaging effect, which can be regarded as the law of large numbers over time. Suppose $\Delta t$ is the time step. When integrating with proper thermostats and barostats, such as Langevin dynamics~\cite{Frenkel2001Understanding}, strong convergence and geometric ergodicity have been established rigorously~\cite{jin2020random,jin2022random,jin2023ergodicity} under some regular assumptions. The error estimates are expressed as:
\begin{equation}
\sup_{t \in[0, T]} \sqrt{\frac{1}{N} \sum_{i=1}^N \mathbb{E}\left(\left|\boldsymbol{r}_i-\boldsymbol{r}_i^*\right|^2+\left|\boldsymbol{p}_i-\boldsymbol{p}_i^*\right|^2\right)} \lesssim \sqrt{\Lambda \Delta t},
\end{equation}
where $\bm{p}_i$ denotes the momentum, $(\bm{r}_i^*,\bm{p}_i^*)$ represent conjugate variables in random batch-based dynamics, and $\Lambda$ is the upper bound for the variance. By Theorem~\ref{thm::3.2}, one has $\Lambda\sim O(1/P)$. While simulating the microcanonical (NVE) ensemble, random batch-type methods can be integrated with an additional weak-coupled bath on the Newtonian dynamics to maintain energy stability~\cite{Liang2024Energy}. For further discussions on other types of baths, we refer the reader to~\cite{Jin2020SISC,Liang2023SISC}.

Finally, we analyze the complexity of the proposed RBSOG method at each time step. By introducing a neighbor list in calculations, the complexity of the near-field calculations is $O(N)$. And by the adaptive random batch importance sampling strategy, the cost of Fourier space calculations is $O(PN)$. This implies that the RBSOG method has linear complexity per time step if one chooses $P=O(1)$.

\section{Numerical Examples}\label{sec::numexa}
In this section, we provide several numerical examples to examine the accuracy and performance of the RBSOG method. The SOG decomposition uses the parameters listed in the second row of Table~\ref{tab::1.1} such that the error is at the level of $10^{-4}$. Our code is developed based on a modification of the LAMMPS software~\cite{thompson2022lammps} (version 21Nov2023). All the calculations were performed on the ``Siyuan Mark-I'' cluster at Shanghai Jiao Tong University, which comprises $2~\times$  Intel Xeon ICX Platinum 8358 CPU ($2.6$ GHz, $32$ cores) and $512$ GB memory per node. 

\subsection{Accuracy for YOCP systems}
In YOCP systems, the static and dynamics properties are characterized by an effective coupling parameter
\begin{equation}
\Gamma^*=\widetilde{\Gamma}(1+\kappa+\kappa^2)e^{-\kappa},
\end{equation} 
where $\widetilde{\Gamma}=Q^2/(ak_{\text{B}}T)$ is the coupling parameter, and $\kappa=a/{\lambda}$ is the screening parameter. Here, $a=(4\pi n_c)^{-1/3}$ represents the average interparticle distance, with $n_c=N/V$ being the charge number density. We conduct MD simulations on a YOCP system with $20000$ particles of charge $Q=1e$ in a cubic box of side length $20~\mathring{A}$. The RBSOG method is compared with our self-implemented PPPM method for the Yukawa potential in the LAMMPS, using the framework and parameter selection scheme outlined in~\cite{dharuman2017generalized}. This implementation closely follows the established PPPM method for the pure Coulomb kernel in the LAMMPS~\cite{thompson2022lammps}, ensuring comparable performance and scalability. The time step of both methods is $\Delta t=10^{-3}$ $fs$. The simulations start with $10^5$ equilibration steps in the NVT ensemble with Nos\'e-Hoover dynamics~\cite{Frenkel2001Understanding}, followed by $5\times 10^5$ production steps in the NVE ensemble. In NVE simulations, the PPPM and the RBSOG are integrated with the sympletic velocity-Verlet method~\cite{Frenkel2001Understanding} and the weakly-coupled scheme~\cite{Liang2024Energy}, respectively. The screening length is $\lambda=0.5773~\mathring{A}$, and the real-space cutoff is set as $r_c=2.8865~\mathring{A}$ for both the RBSOG and the PPPM methods. The estimated relative error level for the PPPM is set to $10^{-4}$, consistent with the estimated level of the SOG decomposition. 

\begin{figure}[!ht]
	\centering
\includegraphics[width=0.88\linewidth]{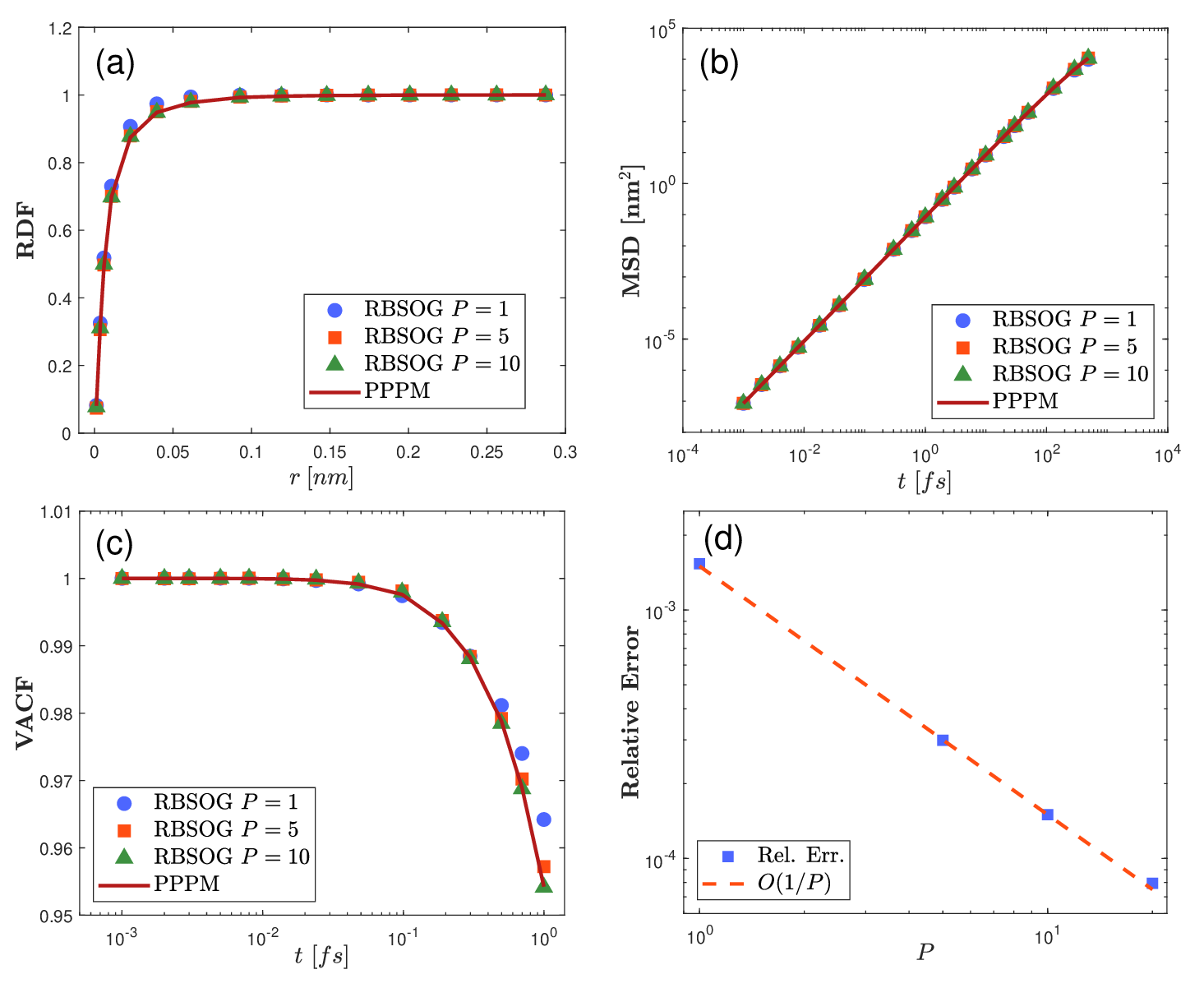}
	\caption{The RDF (a), MSD (b), VACF (c) and the relative error on the ensemble average of total energy (d) of a weak coupling system ($\Gamma^*=0.1$). The simulation results use the RBSOG method with different batch sizes $P=1$, $5$ and $10$, compared to the PPPM with $10^{-4}$ accuracy. }
	\label{fig:01}
\end{figure}

We measure accuracy using four physical quantities: the radial distribution function (RDF), mean square displacement (MSD), velocity auto-correlation function (VACF), and relative error in the ensemble average of total energy. Figures~\ref{fig:01} and~\ref{fig:02} present the simulation results for weak ($\Gamma^* = 0.1$) and strong ($\Gamma^* = 100$) coupling systems, respectively. The results indicate that the RBSOG method with $P = 10$ and $P = 20$ yields statistically identical RDFs, MSDs, and VACFs for $\Gamma^* = 0.1$ and $\Gamma^* = 100$ cases when compared to the PPPM method. The convergence of the total energy in Figures~\ref{fig:01}(d) and~\ref{fig:02}(d) shows an $O(1/P)$ rate, in agreement with the theoretical estimates.

The static structure factor (SSF) $S(k)$ is a critical quantity that characterizes the average structural information of the system~\cite{Allen2017ComputerLiquids}, and is defined as
\begin{equation}
S(k)=\frac{1}{N}\left\langle\rho_{\text{static}}(\bm{k})\rho_{\text{static}}(-\bm{k})\right\rangle\quad\,\hbox{with}\quad\,\rho_{\text{static}}(\bm{k}) = \sum\limits_{j=1}^{N} e^{-i\bm{k}\cdot\bm{r}_j},
\end{equation}
where $\rho_{\text{static}}(\bm{k})$ represents the Fourier transform of particle distributions. The results in Figures~\ref{fig:01} and~\ref{fig:02} suggest that a batch size of \( P = 20 \) is likely sufficient for handling YOCP systems across all coupling factors. To verify this, we conduct simulations to calculate the RDF, MSD, VACF, and SSF for \(\Gamma^*\) ranging from 0.1 to 100, displayed in Figure~\ref{fig:03}. The RBSOG method with \( P = 20 \) demonstrates good agreement with the reference results, indicating its accuracy across the entire range of \(\Gamma^*\).

\begin{figure}[!ht]
	\centering
\includegraphics[width=0.88\linewidth]{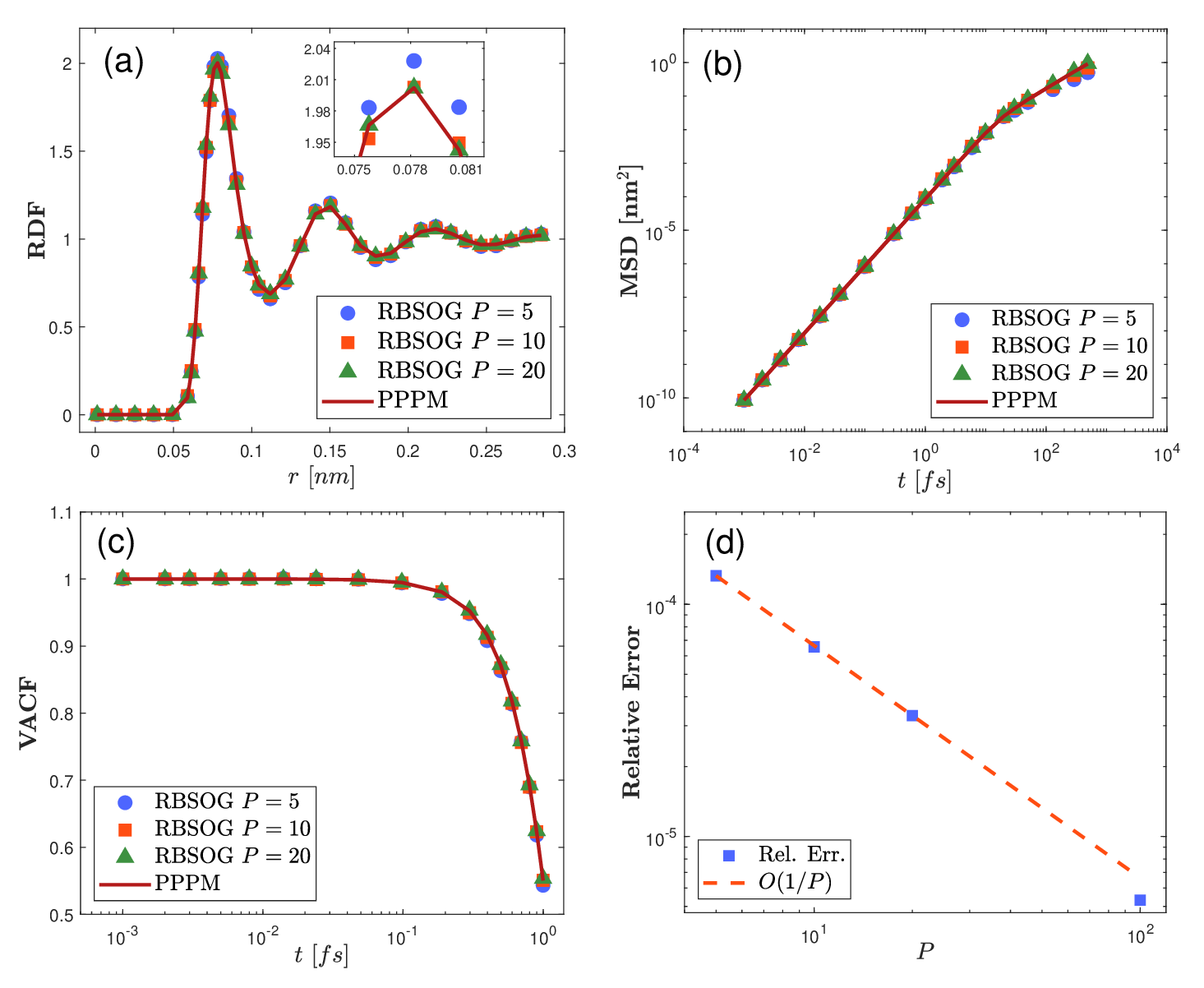}
	\caption{The RDF (a), MSD (b), VACF (c) and the relative error on the ensemble average of total energy (d) of a strong coupling system ($\Gamma^*=100$). The simulation results use the RBSOG method with different batch sizes $P=5$, $10$ and $20$, compared to the PPPM with $10^{-4}$ accuracy.}
	\label{fig:02}
\end{figure}

\begin{figure}[!ht]
\centering
\includegraphics[width=0.88\linewidth]{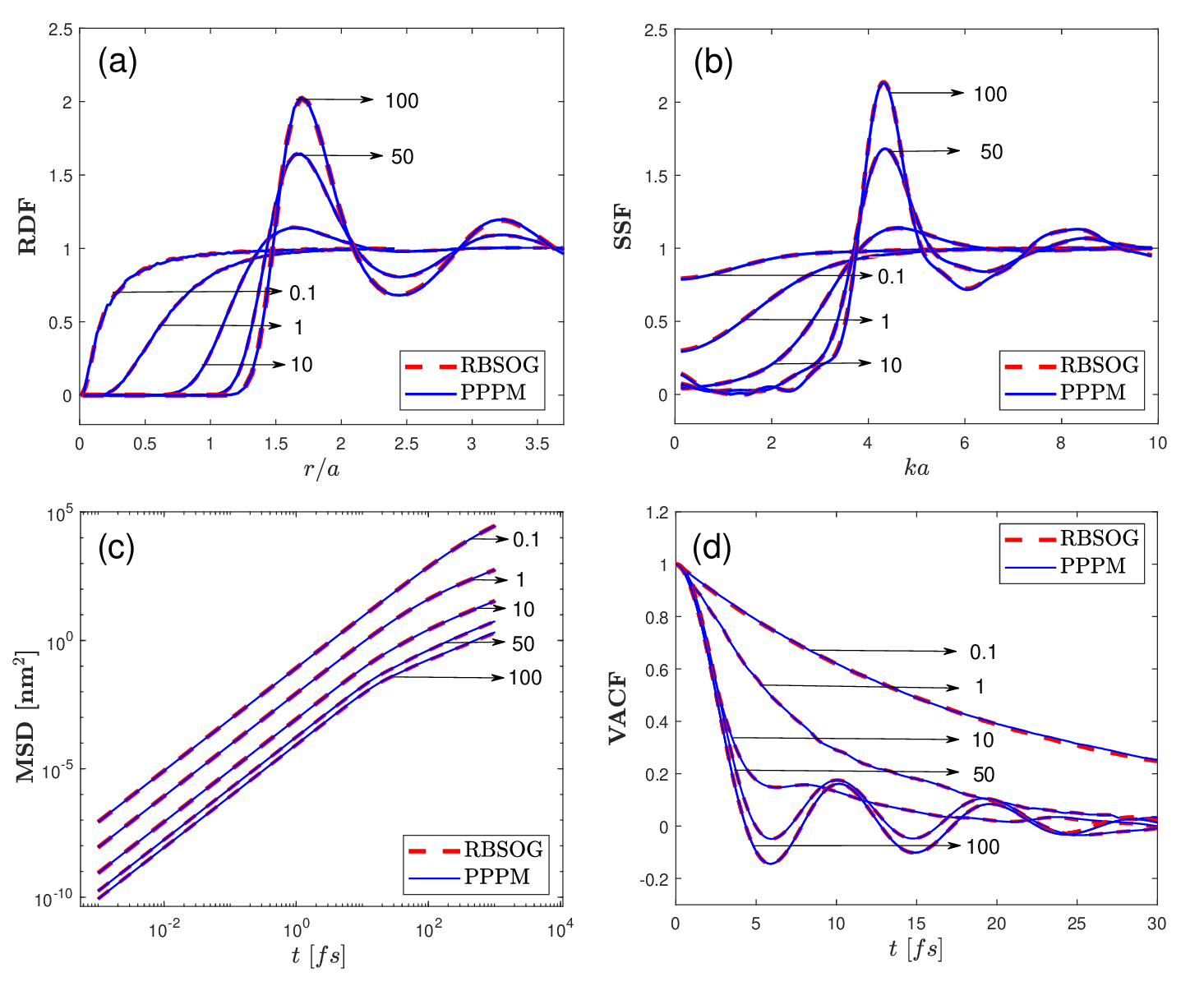}
\caption{The RDF (a), SSF (b), MSD (c) and VACF (d) produced by the PPPM (blue solid lines) and the RBSOG method (red dashed lines) with $P=20$. Data are shown for various effective coupling parameter $\Gamma^*$, ranging from weak ($\Gamma^*=0.1$) to strong ($\Gamma^*=100$) coupling regimes. }
	\label{fig:03}
\end{figure}

\subsection{Wall-clock time performance}
We now investigate the CPU performance by comparing the results of the RBSOG, PPPM and PVFMM. For the RBSOG and PPPM, we use our LAMMPS implementations. For the PVFMM, we use the open source libraries~\cite{Malhotra2016PVFMM}, where the Scientific computing template library (SCTL)~\cite{Malhotra2022SCTL} is used for the SIMD accelerated kernel evaluation. 
It is worth noting that, although some periodic FMMs have been developed~\cite{jiang2023jcp,yan2018flexibly}, to the best of our knowledge, no open-source software currently supports fully periodic 3D Yukawa systems (though some of them claim kernel-independent in the formulation). Consequently, we compare our method with the PVFMM where the periodic tiling is direct truncated. To access a fair comparison, the estimated relative force error and real-space cutoff are set as $10^{-4}$ and $2.8865~\mathring{A}$ for the RBSOG and PPPM, respectively, where the real space cost is roughly identical for both methods. For the PVFMM, the multipole expansion order is set to $5$ and the maximum number of points in a leaf node is set to $50$ for $\varepsilon=10^{-4}$, and the periodic tiling is truncated at the same accuracy. The main goal of such parameter choice is for solely comparing the improvement of the RBSOG in Fourier space. We expect a fine tuning of parameters such as the cutoff $r_c$ and the batch size $P$ to balance the cost of the RBSOG in real and Fourier spaces can further optimize its efficiency in practice. All the simulations were conducted for $10000$ steps to estimate the CPU time per step. 

We first present the computational complexity of the proposed RBSOG method. In Figure~\ref{fig:07}(a), 512 cores are used for simulating strongly coupled YOCP systems with the screening length $\lambda=0.5773~\mathring{A}$, $\Gamma^*=100$ and number density $ 2.5~\mathring{A}^{-3}$, and the computational time per step is shown for system sizes up to $N=1.28\times 10^6$. The dashed lines indicate linear fitting. The results demonstrate the \(O(N)\) scaling of the RBSOG method. The initial few data points of the Fourier space component do not scale linearly due to the small average number of particles per core and the dominance of communication costs. Next, we examine the memory usage with an increasing number of CPU cores. Figure~\ref{fig:07}(b) shows memory allocation per MPI rank while simulating the same YOCP system above. Compared to the PPPM and PVFMM results, the RBSOG method significantly reduces memory usage by about $40\%$ due to its tree-free and mesh-free nature. These findings highlight the potential for broader applications of the RBSOG method in large-scale simulations.

\begin{figure}[!ht]
	\centering
\includegraphics[width=0.87\linewidth]{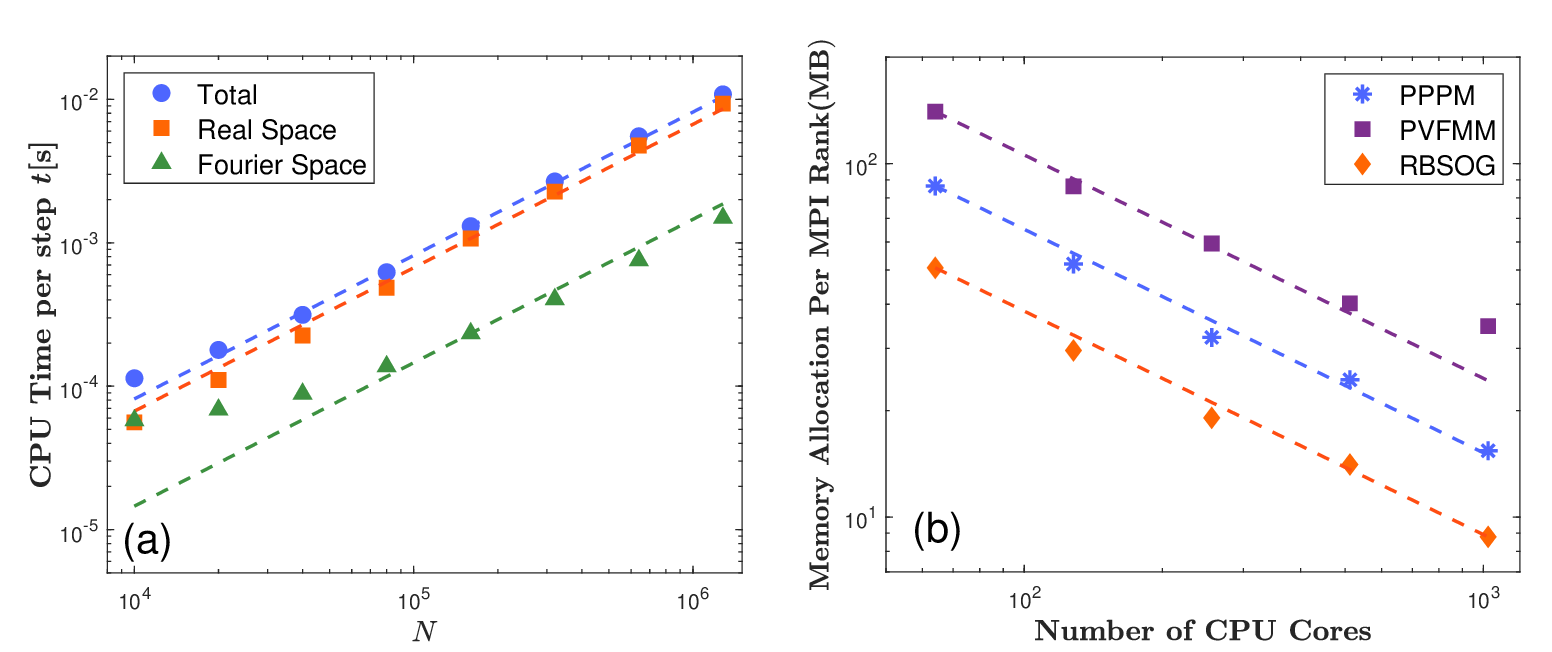}
	\caption{(a) CPU time per step for the RBSOG method with increasing $N$ and (b) memory allocation per MPI rank for the RBSOG, PPPM and PVFMM methods with increasing number of CPU cores. In (a), data are shown for the real space and Fourier space parts, as well as the total (i.e., real~+~Fourier) CPU cost. The dashed lines in (a-b) show the linear fitting of data.}
	\label{fig:07}
\end{figure}

Figure~\ref{fig:04}(a-b) presents the results for CPU time per step and scalability in strong scaling tests. Strong scaling measures parallel performance as the number of CPU cores increases while keeping the system size fixed. For these tests, we use a YOCP system with $1.28 \times 10^6$ particles and a side length of $108.57~\mathring{A}$. When a small number of cores is used, the RBSOG method outperforms the PPPM and PVFMM by a factor of $2-3$. As the number of cores increases, the RBSOG achieves an order of magnitude improvement over the other two methods when $\mathcal{C}\geq100$ CPU cores are used. Moreover, the RBSOG maintains nearly $95\%$ parallel scalability even with $1024$ CPU cores. Note that the PVFMM also has good scalability when $\mathcal{C}\leq 64$ and outperforms the PPPM throughout the tests. Although the time and memory cost of the PVFMM are not advantageous in our tests, we conjecture that they could be improved if combined with an appropriate periodization method. 

The results for CPU time per step and scalability in weak scaling tests are presented in Figure~\ref{fig:04}(c-d). Weak scaling measures how the solution time changes with the number of processors while maintaining a fixed average number of particles per processor. We conduct tests using up to $1024$ cores for YOCP systems with same number density as in the strong scaling tests. The RBSOG achieves near-perfect weak scaling, even with a relatively small average of $2000$ particles per core. In contrast, the strong and weak scaling of the PPPM and the PVFMM drop to about $10\%$ and $25\%$, respectively, when $1024$ cores are used. These results demonstrate the promising parallel efficiency of the RBSOG method. 

\begin{remark}
More recently, the dual-space multilevel kernel-splitting (DMK) framework~\cite{jiang2024CPAM} has emerged, showing promise for greater improvements as an alternative to both the FMM and fast Ewald summation. However, its current implementation is restricted to a serial Fortran version and does not support periodic boundary conditions. As the single-core efficiency of the DMK is comparable to that of the PVFMM, we limit our comparison to the PVFMM in this study.  
\end{remark}

\begin{figure}[!ht]
	\centering
\includegraphics[width=0.92\linewidth]{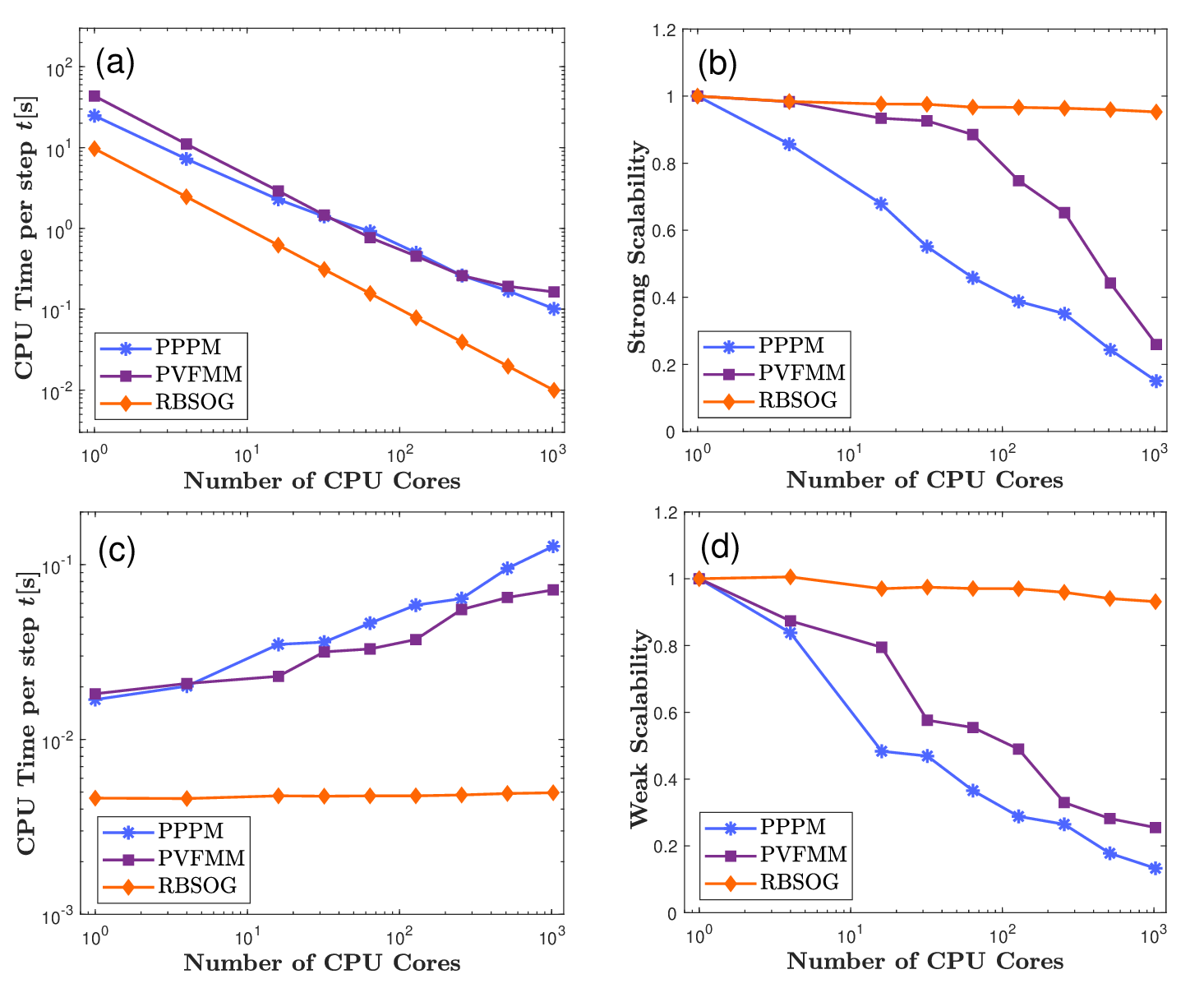}
	\caption{CPU time and strong/weak scalability are compared between the RBSOG, PPPM and PVFMM using up to \(1024\) CPU cores. (a-b) present the results for a fixed $N= 1.28 \times 10^6$ with the increase of CPU cores. (c-d) present the results with an average of $2000$ particles per core.}
	\label{fig:04}
\end{figure}

\begin{figure}[!ht]
\centering	
\includegraphics[width=0.87\linewidth]{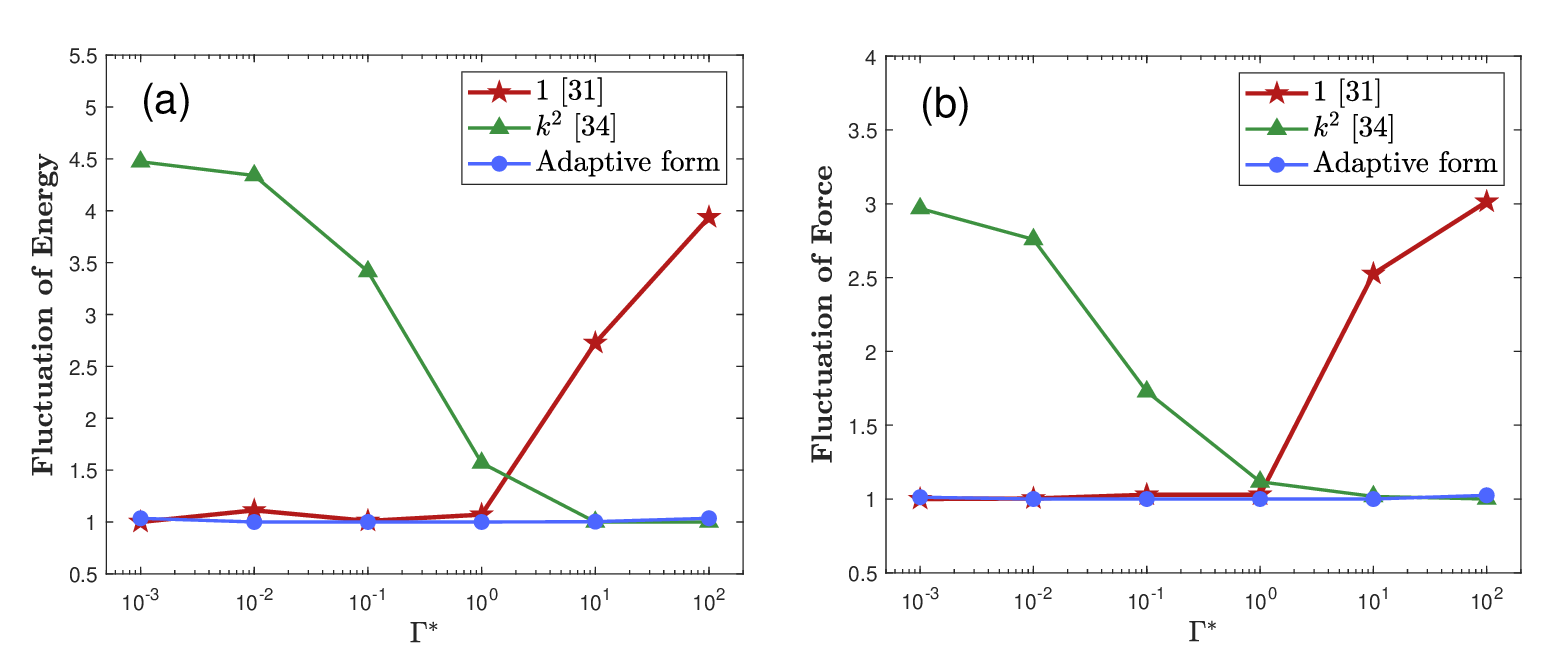}
\caption{The average fluctuations in (a) energy and (b) force are analyzed across different effective coupling factors $\Gamma^*$ during a $10^5$ time step simulation of YOCP systems. To better observe relative levels, we normalized the lowest fluctuation to $1$. The simulations use the RBSOG method with varying corrections $\mathscr{F}(\bm{k})$ for the sampling measure in Eq.~\eqref{eq::3.124}, with a fixed batch size of $P=20$.}
\label{fig:05}
\end{figure}

\subsection{Variance reduction}\label{subsec::variancereduction}
In this section, we compare the adaptive importance sampling strategy from Section~\ref{subsec:important} with different strategies by examining the average energy and force fluctuations. These quantities, defined as 
\begin{equation}
E_{\text{fluc}}=\langle (E^*-E)^2\rangle\quad\text{and}\quad F_{\text{fluc}}=\left\langle \frac{1}{N}\sum_{i=1}^N\left(\bm{F}_i^*-\bm{F}_i\right)^2 \right\rangle,
\end{equation}
measure the variance of energy and forces in ensemble averages, where $E^*$ and $\bm{F}^*_i$ are the stochastic energy and force (exerted at $i$th particle) calculated by the RBSOG method. We conduct MD simulations for $10^5$ time steps on YOCP systems of $1000$ particles at various of effective coupling factors $\Gamma^*$, where the screening length and real-space cutoff are set as $\lambda=0.5773~\mathring{A}$ and $r_c=2.8865~\mathring{A}$, respectively. The batch size is fixed at $P=20$ for all cases.

In Figure~\ref{fig:05}, we present the results of $E_{\text{fluc}}$ and $F_{\text{fluc}}$, where we test three different correction choices for $\mathscr{F}(\bm{k})$ to approximate the energy and forces: $\mathscr{F}(\bm{k})\equiv 1$, $\mathscr{F}(\bm{k})=k^2$, and the adaptive version from Eqs.~\eqref{eq::3.15}-\eqref{eq3::3.22}. The first two options have been used in the $\lambda \rightarrow \infty$ (Coulomb) case for previous RBE~\cite{Jin2020SISC} and RBSOG~\cite{Liang2023SISC} methods. In Yukawa systems, we observed that variance fluctuations due to different sampling strategies vary with the effective coupling parameter $\Gamma^*$. For small $\Gamma^*$, the long-range correlations from the Yukawa potential are minimal, suggesting that not correcting the structure factor ($\mathscr{F}(\bm{k})\equiv 1$) is preferable. Conversely, for large $\Gamma^*$, increased long-range correlations cause $\langle|\rho(\bm{k})|^2\rangle$ to approach $k^2$ due to the Stillinger-Lovett condition~\cite{Lovett1968JCP}. Across a wide range of coupling strengths ($\Gamma^*$ from $0.001$ to $100$), our adaptive importance sampling strategy consistently achieves the lowest variance among the three methods, reducing variance by up to $3-4$ times compared to other strategies. This efficiency allows for using only $1/3$ to $1/4$ of the batch size to achieve comparable performance. This advantage is anticipated to be even more pronounced in non-equilibrium systems with $\Gamma^*$ varying during the whole simulations.

\begin{figure}[!ht]
\centering	
\includegraphics[width=0.87\linewidth]{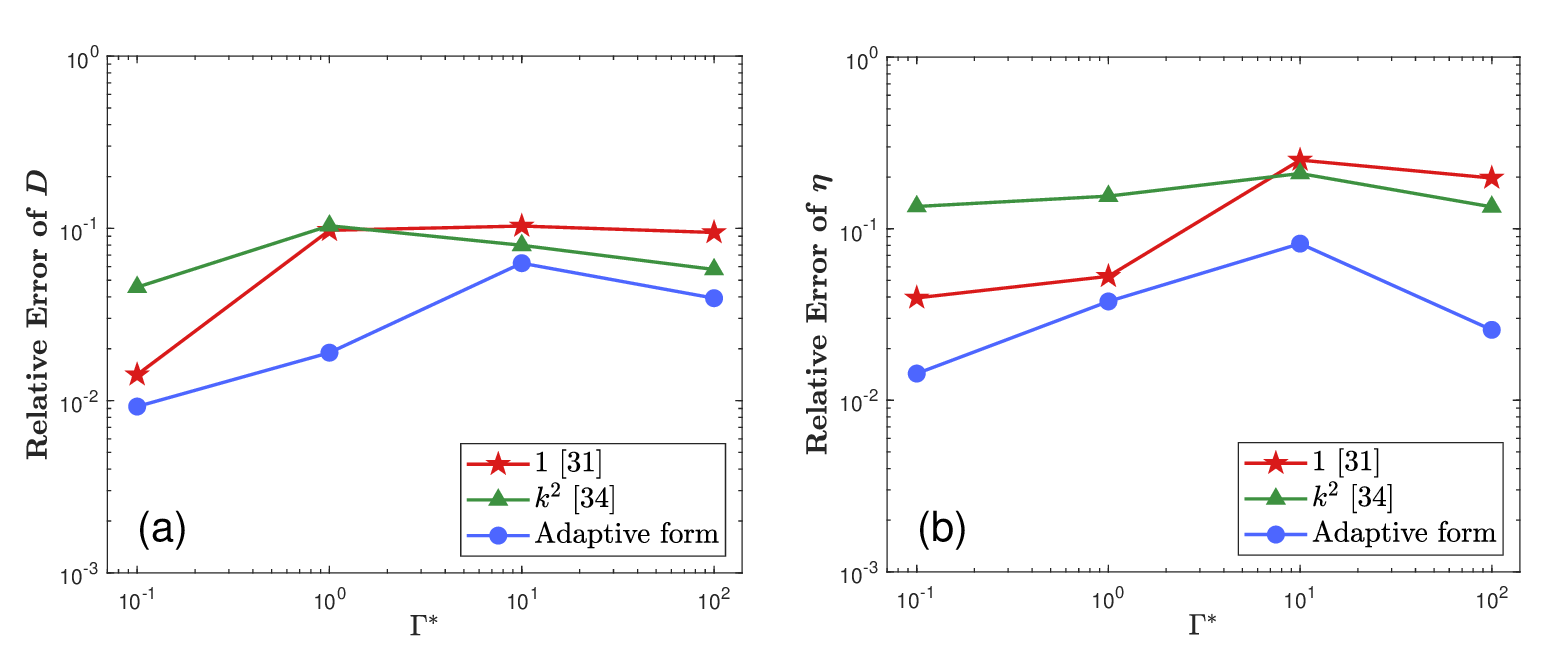}
\caption{The relative errors of (a) the self-diffusion coefficient $D$ and (b) the thermal conductivity $\eta$ across different effective coupling factors $\Gamma^*$ during a $10^5$ time step simulation of YOCP systems. The simulations use the RBSOG method with varying corrections $\mathscr{F}(\bm{k})$ on the sampling measure as defined in Eq.~\eqref{eq::3.124}, using a fixed batch size of $P=20$.}
\label{fig::SDCTC}
\end{figure}

To further evaluate the advantages of our adaptive importance sampling strategy compared to those in~\cite{Jin2020SISC,Liang2023SISC}, we calculate the self-diffusion coefficient and thermal conductivity, two key quantities characterizing accuracy on spatial dynamics and thermodynamics. The self-diffusion coefficient $D$ is calculated from the Einstein relation~\cite{Frenkel2001Understanding}:
\begin{equation}
D=\lim_{t\to\infty}\frac{1}{6t}\langle\vert\bm{r}_i(t)-\bm{r}_i(0)\vert^2\rangle,
\end{equation}
where $\bm{r}_i(t)$ represents the position of the $i$th particle at time $t$. The thermal conductivity $\eta$ is calculated using the Green-Kubo relation: 
\begin{equation}
\eta=\lim_{\tau\to\infty}\frac{1}{k_BT^2V}\int_0^{\tau}\langle\bm{J}(0)\cdot \bm{J}(t)\rangle dt,
\end{equation}
where $t$ is the time and the heat flux $\bm{J}(t)$ is defined as: 
\begin{equation}
\bm{J}(t)=\sum_{i=1}^N\left[\frac{1}{2}m_{i}\vert\bm{v}_i(t)\vert^2\bm{v}_i(t)+\frac{1}{2}\sum_{j\neq i}^N\phi(|\bm{r}_{ij}(t)|)\bm{v}_i(t)+\frac{1}{2}\sum_{j\neq i}^N\left(\bm{r}_{ij}(t)\cdot\bm{v}_i(t)\right)\bm{F}_{ij}(t)
 \right],
\end{equation}
with $\phi(\cdot)$ representing the potential function. The results in Figure~\ref{fig::SDCTC}(a-b) demonstrate that our adaptive importance sampling strategy achieves the best accuracy, outperforming the other two sampling strategies from~\cite{Jin2020SISC} and \cite{Liang2023SISC} across the entire range of the effective coupling factor $\Gamma^*$. This further highlights the robustness and effectiveness of our proposed method in accurately capturing dynamic properties for Yukawa systems.

\subsection{Application to the deuterium-\texorpdfstring{$\alpha$}~ mixture}\label{subsec::energyconservation}

Measuring the input energy from $\alpha$-heating is critical for achieving fusion ignition~\cite{betti2015alpha}. In this process, the energy from $\alpha$ particles produced by fusion is deposited in the fusion plasma (deuterium), transferring energy to it. Due to the high-temperature and high-density nature of the system, it is hard to simulate such system using MD, which requires a fairly small time step $\Delta t$, leading to substantial computational cost~\cite{graziani2012large}. Our proposed method provides a promising solution to this issue.

We consider a deuterium-$\alpha$ mixture with $45200$ deuterium particles and $2600$ $\alpha$ particles in a cubic cell with a side length of $10\mathring{A}$. The initial system, a deuterium plasma, is in thermal equilibrium in the NVT ensemble at a temperature of $T_{\text{D}}=3$ keV and a number density of $n_{\text{D}}=45.2~\mathring{A}^{-3}$. We then add the high-energy $\alpha$ particles at a temperature of $T_{\alpha}=3.54$ MeV and perform simulations in the NVE ensemble. 

In the simulation, the PPPM and the RBSOG are integrated with the sympletic velocity-Verlet method~\cite{Frenkel2001Understanding} and the weakly-coupled scheme~\cite{Liang2024Energy}, respectively, to maintain energy stability. 
Figure~\ref{fig:06}(a) shows the energy evolution for the RBSOG and the PPPM methods with different time steps $\Delta t$. The RBSOG with $\Delta t=10^{-4}$ $fs$ and the PPPM with $\Delta t=2\times10^{-6}$ $fs$ maintain energy stability for $100$ $fs$ simulations, while significant energy drift is observed for the PPPM with $\Delta t=10^{-4}$ and $10^{-5}$ $fs$. Occasional small energy fluctuations with the RBSOG are due to large-angle scattering of closely interacting particles, but our method corrects back to the reference value in a quick time. This suggests that the RBSOG can provide stable and accurate results with a large time step.

Physically, when a hot $\alpha$ particle interacts with a relatively cold deuterium particle, the hot $\alpha$ particle transfers a large amount of energy to the deuterium particle, causing the temperature of the $\alpha$ particle to drop. In Figure~\ref{fig:06}(b), we present the evolution of the temperature of $\alpha$ particles. Both the RBSOG with $\Delta t=10^{-4}$ $fs$ and the PPPM with $\Delta t=2\times 10^{-6}$ $fs$ explicitly capture this cooling process for a period of $10$ $fs$. However, due to the energy drift, this energy exchange process is not obvious for the PPPM with $\Delta t=10^{-5}$ $fs$. Even more, the PPPM with $\Delta t=10^{-4}$ $fs$ incorrectly shows a warming of $\alpha$ particles. This further demonstrates that the RBSOG with an appropriate energy bath can accurately capture the physical properties of high-temperature and high-density plasma systems using larger time steps, which is difficult to achieve with existing mainstream algorithms. 

\begin{figure}[!ht]
\centering	\includegraphics[width=0.87\linewidth]{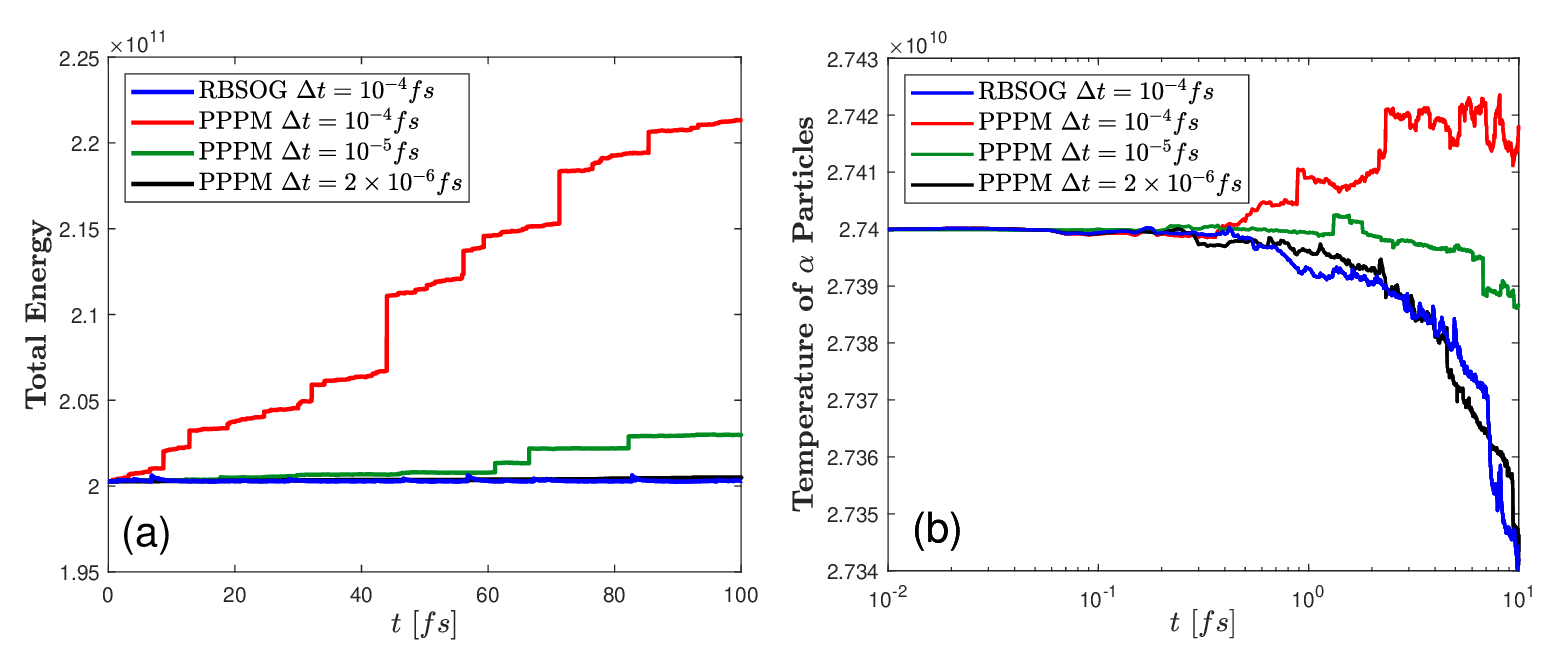}
\caption{Simulation results of high-temperature, high-density deuterium-$\alpha$ mixtures. (a): Evolution of total energy. (b): Evolution of temperature of $\alpha$ particles. Data are shown for the RBSOG and the PPPM with different time steps $\Delta t$.}
\label{fig:06}
\end{figure}

\section{Concluding remarks}\label{sec::conclu}
In summary, we have introduced a novel RBSOG method for efficiently simulating three-dimensional Yukawa systems. This approach is based on a new SOG decomposition of the Yukawa kernel, providing improved accuracy and regularity over the traditional Ewald decomposition. By employing the idea of random mini-batch in Fourier space with an adaptive importance sampling strategy, our algorithm achieves \(O(N)\) complexity, high parallel scalability, and near-optimal variance reduction across all coupling parameters. We provide rigorous analysis on the SOG decomposition construction, variance reduction, and simulation convergence. Numerical simulations of YOCP systems with both weak and strong coupling demonstrate the accuracy, efficiency, and scalability of our method. Compared to the PPPM and PVFMM methods, our approach accelerates simulations by an order of magnitude using \(1.28 \times 10^6\) charges and 1024 cores. Simulations in high-temperature, high-density deuterium-$\alpha$ mixtures highlight the potential of our method for applications in fusion ignition plasma systems. Furthermore, the RBSOG method can be easily extended to other interacting kernels in plasma simulations using kernel-independent SOG approximations~\cite{greengard2018anisotropic,gao2021kernelindependent}. Future work will focus on extending the method to quasi-2D Yukawa systems, addressing challenges such as confinement effects and dielectric mismatches~\cite{gan2024fast,gan2024random}. 

\appendix

\section{Proof of Lemma \ref{lemma::1.4}}\label{app::A.1}
Let $\mathscr{R}(\nu)$ denote the real part of $\nu$. Under the conditions $\mathscr{R}(\nu)>-1/2$, $x>0$ and $z>0$, the modified Bessel function $K_{\nu}(xz)$ has the Basset's integral representation \cite{abramowitz1964handbook}
\begin{equation}
	K_{\nu}(xz)=\frac{\Gamma(\nu+\frac{1}{2})(2z)^{\nu}}{\sqrt{\pi}x^{\nu}}\int_{0}^{\infty}\frac{\cos(xt)}{(t^2+z^2)^{\nu+\frac{1}{2}}}dt,
\end{equation}
where
\begin{equation}
	\Gamma(x):=\int_{0}^{\infty}t^{x-1}e^{-t}dt
\end{equation}
is the Gamma function. It follows that
\begin{equation}\label{eq::A.3}
	\begin{split}
	\left|K_{\frac{1}{2}-2\pi ik}\left(\frac{r}{\lambda}\right)\right|=\left|\Gamma(1-2\pi ik)\right|\sqrt{\frac{2\lambda}{\pi r}}\left|\int_{0}^{\infty}\frac{\cos\left(\frac{rt}{\lambda}\right)}{(1+t^2)^{1-2\pi ik}}dt\right|\leq \sqrt{\frac{\pi\lambda}{2r}}\left|\Gamma(1-2\pi i k)\right|.
	\end{split}
\end{equation}
The Gamma funtion on the right-hand side has a useful inequality \cite{cuyt2008handbook}
\begin{equation}\label{eq::A.4}
	\left|\Gamma(1-2\pi i k)\right|\leq (2\pi)^{\frac{1}{2}}e^{-\pi^2k}e^{\frac{1}{6|1-2\pi i k|}}.
\end{equation}
Combining Eqs.~\eqref{eq::A.3} and \eqref{eq::A.4}, we finish the proof.

\section{Analysis of the variance of force}\label{sec::energy}
Recall the sampling measure given by Eq.~\eqref{eq::3.124}, where the correction $\mathscr{F}(\bm{k})$ is taken in the adaptive form of Eq.~\eqref{eq3::3.22}. Let us study the variance of force in this case. By Eq.~\eqref{eq3::3.22}, it can be derived that
\begin{equation}
	\frac{k^2V}{\langle\left|\rho(\bm{k})\right|^2\rangle}=\frac{k_{\text{B}}Tk^2/Q+1/\epsilon(\bm{k})}{k_{\text{B}}T}\leq \frac{1}{k_{\text{B}}T}+k^2/Q:=D(k)
\end{equation}
where we use the fact that $\epsilon(\bm{k})^{-1}\leq 1$~\cite{hansen2013theory}. By this inequality and Eq.~\eqref{eq::3.29}, one has
\begin{equation}\label{eq::3.28}
	\begin{split}
		\mathbb{E}|\bm{\Xi}_{F,i}|^2
		&\leq \frac{1}{P}\dfrac{C Shq_{i}^2}{4\pi^2 V}\sum_{m=-M_1}^{M_2}e^{-t_m}\int_{0}^{\infty}4\pi k^2e^{-e^{-t_m}k^2/4}D(k)dk\\
		&=\frac{1}{P}\dfrac{2CShq_{i}^2}{\sqrt{\pi} V}\sum_{m=-M_1}^{M_2}\left(\frac{e^{t_m/2}}{k_{\text{B}}T}+\frac{6e^{3t_m/2}}{Q}\right)\\
		&\leq \frac{1}{P}\dfrac{2C_1Shq_{i}^2}{\sqrt{\pi} V},
	\end{split}
\end{equation}
where constant $C_1/C$ arises from the fact that the sum over $m$ in the third equation is bounded. By the definition of normalization factor $S$, one has the following estimate:
\begin{equation}
	\begin{split}
		S
  \leq \frac{hQV}{2\pi}\sum_{m=-M_1}^{M_2}e^{-t_m}\int_{0}^{\infty}k^2e^{-e^{-t_{m}}k^2/4}dk=\frac{hQV}{\sqrt{\pi}}\sum_{m=-M_1}^{M_2}e^{t_{m}/2}.
	\end{split}
\end{equation}
Since $\sum\limits_{m=-M_1}^{M_2}e^{t_{m}/2}$ is bounded by an $O(1)$ constant $C_2$, one has 
\begin{equation}\label{eq::3.30}
	S\leq C_2 \pi^{-1/2}h\rho \max_{i}|q_i|^2V= O(V).
\end{equation} 
Substituting Eq.~\eqref{eq::3.30} into Eq.~\eqref{eq::3.28} gives $\mathbb{E}|\bm{\Xi}_{F,i}|^2=O\left(P^{-1}\right)$ which is also independent of both the particle numbers and the number of Gaussians.

\section{Parameter sets for \texorpdfstring{$C^2$}~-continuous SOG decomposition}\label{app::C2continuous}

In Table~\ref{tab::1.2}, we provide parameter sets for $C^2$-continuous SOG decomposition, analogous to those presented in Table~\ref{tab::1.1} for $C^1$-continuous decomposition. This is done by write the far-field part of SOG decomposition as
\begin{equation}\label{eq::c2}
 	\mathcal{F}_{h}^{t_0}(r)=h\left[w_{M_2}\dfrac{1}{\sqrt{\pi}}e^{-r^2e^{t_{M_2}}s_{M_2}-\frac{1}{4\lambda^2e^{t_{M_2}}}+t_{M_2}/2}+\sum_{m=-M_1}^{M_2-1}f\left(t_m,r\right)\right],
\end{equation}
and conjointly solve the potential continuity condition Eq.~\eqref{eq::2.38}, the force continuity condition Eq.~\eqref{eq::1.25}, and the second derivative condition 
\begin{equation}
\frac{d^2}{dr^2}\left[Y(r)-\mathcal{F}_{h}^{t_0}(r)\right]{\bigg|}_{r=r_c}=0
\end{equation}
for $t_0$, $\omega_{M_2}$, and $s_{M_2}$.

\renewcommand\arraystretch{1.3}
\begin{table}[!ht]
\caption{Parameter sets for $C^{2}$-continuous SOG decomposition. $M_{\text{tot}}:=M_1+M_2+1$ is the minimal number of Gaussians satisfying the error criteria on the region $[r_c,R]$ with $r_c=5\lambda$, $R=33\lambda$, and $\lambda=0.5773$.}
\centering
\scalebox{0.88}{
\begin{tabular}{cccccccc}
\hline
$\epsilon$ &$h$ &$ t_0 $ &$M_1$ &$M_2$ & $M_{\text{tot}}$ &$\omega_{M_2}$  &$s_{M_2}$ \\ \hline
		$10^{-3}$    &1.131155934143089 &0.320762776093483  &3  &0 & 4 &0.989315895934884   &0.925544356078519 \\ \hline
		$10^{-4}$   &0.894984933518395 &0.503827089686760   &4   &0 & 5 &0.991286398776329   &0.916780172492483 \\ \hline
		$10^{-5}$ 	  &0.740391708745519 &0.199987634558998   &5   &1 & 7  &1.011505830491367 &0.998238561838385 \\  \hline
		$10^{-7}$  	&0.550285792019561 &0.105931467073842   &7   &1  & 9 &1.010000002048907   &1.053279452264880 \\  \hline
		$10^{-9}$ 	  &0.437859267899280 &0.159907820763955   &9   &1 & 11 &0.990000003971023   &0.856307821865310 \\  \hline
		$10^{-11}$    &0.363578174148321 &0.169998676580835   &12   &2 
&  15 &1.000000000000379   &1.003207791668459 \\ \hline
		$10^{-13}$ 	&0.310844614243983 &0.096999894851203   &13   &3 
 & 17  &1.000000000000002   &0.998956217641487 \\  
		\hline
	\end{tabular}}
\label{tab::1.2}
\end{table}

\section*{Acknowledgments}
This work was supported by the National Natural Science Foundation of China (Grants No. 12325113, 12401570 and 12426304) and the Science and Technology Commission of Shanghai Municipality (Grant No. 23JC1402300). The work of J. L. is partially supported by the China Postdoctoral Science Foundation (grant No. 2024M751948). The authors would like to thank the support from the Center for High Performance Computing at Shanghai Jiao Tong University and SJTU Kunpeng \& Ascend Center of Excellence.


\end{document}